\documentclass[10pt,twocolumn,twoside]{IEEEtran}
\usepackage{cite}
\usepackage[utf8]{inputenc} 
\usepackage[T1]{fontenc}    
\usepackage{hyperref}       
\usepackage{url}            
\usepackage{booktabs}       
\usepackage{amsfonts}       
\usepackage{nicefrac}       
\usepackage{microtype}      
\usepackage{doi}
\usepackage{amsmath,amssymb,amsfonts}
\usepackage{bbding}
\usepackage{algorithmic}
\usepackage{graphicx}
\usepackage{textcomp}
\usepackage{xcolor}
\usepackage{bbm}
\usepackage{amsthm}
\usepackage{setspace}
\usepackage{physics}
\usepackage{algorithm}
\usepackage{mathtools}
\usepackage{diagbox}
\usepackage{caption}
\usepackage{subcaption}
\usepackage{physics}
\usepackage{multirow}
\usepackage{epstopdf}
\usepackage{tabularx}
\makeatletter
\newcommand{\distas}[1]{\mathbin{\overset{#1}{\kern\z@\sim}}}%
\newsavebox{\mybox}\newsavebox{\mysim}
\newcommand{\distras}[1]{%
  \savebox{\mybox}{\hbox{\kern3pt$\scriptstyle#1$\kern3pt}}%
  \savebox{\mysim}{\hbox{$\sim$}}%
  \mathbin{\overset{#1}{\kern\z@\resizebox{\wd\mybox}{\ht\mysim}{$\sim$}}}%
}
\makeatother
\makeatletter
\newcommand*\rel@kern[1]{\kern#1\dimexpr\macc@kerna}
\newcommand*\widebar[1]{%
  \begingroup
  \def\mathaccent##1##2{%
    \rel@kern{0.8}%
    \overline{\rel@kern{-0.8}\macc@nucleus\rel@kern{0.2}}%
    \rel@kern{-0.2}%
  }%
  \macc@depth\@ne
  \let\math@bgroup\@empty \let\math@egroup\macc@set@skewchar
  \mathsurround\z@ \frozen@everymath{\mathgroup\macc@group\relax}%
  \macc@set@skewchar\relax
  \let\mathaccentV\macc@nested@a
  \macc@nested@a\relax111{#1}%
  \endgroup
}
\makeatother
\newtheorem{definition}{Definition}
\newtheorem{assumption}{Assumption}
\newtheorem{theorem}{Theorem}
\newtheorem{lemma}{Lemma}
\newtheorem{remark}{Remark}

\newtheorem{corollary}{Corollary}
\newtheorem{proposition}{Proposition}
\newtheorem{property}{Property}
\begin{document}
\title{Distributed Management of Fluctuating Energy Resources in Dynamic Networked Systems}
\author{Xiaotong Cheng$^\dagger$, Ioannis Tsetis$^\dagger$, Setareh Maghsudi\thanks{$\dagger$ Equal contribution. \\
X. Cheng is with the Department of Electrical Engineering and Information Technology, Ruhr-University Bochum, 44801 Bochum, Germany (email:xiaotong.cheng@ruhr-uni-bochum.de). I. Tsetis is with the Department of Computer Science, University of Tübingen, 72074 Tübingen, Germany (email:ioannis.tseitis@uni-tuebingen.de). S. Maghsudi is with the Department of Electrical Engineering and Information Technology, Ruhr-University Bochum, 44801 Bochum, Germany and with the Fraunhofer Heinrich Hertz Institute, 10587 Berlin, Germany (email:setareh.maghsudi@rub.de).\\
A part of this paper appeared at the 2023 IEEE International Conference on Acoustics, Speech, and Signal Processing (ICASSP 2023).}\\
}
\maketitle
\begin{abstract}
Modern power systems integrate renewable distributed energy resources (DERs) as an environment-friendly enhancement to meet the ever-increasing demands. However, the inherent unreliability of renewable energy renders developing DER management algorithms imperative. We study the energy-sharing problem in a system consisting of several DERs. Each agent harvests and distributes renewable energy in its neighborhood to optimize the network's performance while minimizing energy waste. We model this problem as a bandit convex optimization problem with constraints that correspond to each node's limitations for energy production. We propose distributed decision-making policies to solve the formulated problem, where we utilize the notion of dynamic regret as the performance metric. We also include an adjustment strategy in our developed algorithm to reduce the constraint violations. Besides, we design a policy that deals with the non-stationary environment. Theoretical analysis shows the effectiveness of our proposed algorithm. Numerical experiments using a real-world dataset show superior performance of our proposal compared to state-of-the-art methods.
\end{abstract}
%
\section{Introduction}
The beneficial characteristics of renewable energy, such as sustainability, make it inevitable to integrate wind- and solar resources into power generation systems \cite{wei2015distributionally} to satisfy the ever-increasing energy demand while remaining environment-friendly. Distributed energy resources (DERs) are small-scale electricity supply or demand resources connected to the electric grid. They include solar power, wind power, geothermal power, hydrothermal power, etc., as clean energy resources \cite{xu2016toward}. They are a potential environmental- and economically valuable solution for power systems \cite{rahbar2014real}. 

Compared to conventional fossil-based energy generation, DERs reduce carbon dioxide emissions and have lower transmission- and distribution costs \cite{xu2016toward}. Besides, compared with traditional large-scale power plants, DERs are small and highly flexible \cite{wu2016distributed}. Such characteristics enable more power generation and increase supply reliability in a cost-effective manner.

However, the DERs' performance depends on several factors, such as time and weather. As such, they suffer from high uncertainties \cite{xu2016toward,yu2015towards}, which exacerbates their integration into power grids. Besides, the variability, unpredictability, and high complexity of the DERs system remain crucial challenges \cite{wang2015dynamic}. The challenge is even more imminent considering the mismatch between renewable generation and load demand, which causes outages or severe energy waste \cite{rahbar2014real}. Therefore, it is essential to distribute energy generated by DERs carefully and without delay. Furthermore, privacy is crucial, especially on the demand side \cite{wang2013online}. In brief, the challenge is to develop an efficient distribution mechanism in DERs to satisfy users' demands while minimizing energy waste under uncertainty and rare information exchange.

Previous research focused on offline energy resource management \cite{atzeni2013noncooperative, chandy2010simple, zhang2013robust} with the ideal assumption that the generated energy and the load demand are either deterministic or known a priori before scheduling. However, in practical applications, generating renewable resources is time-variant and involves uncertain factors that cannot be predicted or controlled precisely \cite{xu2016toward,yu2015towards,9080557}. To solve the real-time energy distribution problem, \cite{wang2013online} proposes online centralized and distributed algorithms based on the gradient method. Reference \cite{wang2015dynamic} uses the alternating direction method of multipliers (ADMM) within a model predictive control (MPC) framework to control and optimize the power scheduling problem. Reference \cite{rahbar2014real} designs an online algorithm for the real-time energy management of microgrid systems by combining the offline optimal solution with window-based sequential optimization. 

In the prior works that investigate online energy management \cite{wang2013online,rahbar2014real}, online convex optimization (OCO) is a popular framework. However, a large body of literature assumes that the gradient at any point in the decision space is accessible, which often does not hold in real-world applications. To mitigate this shortcoming, we develop a gradient-free online algorithm for energy-sharing in DERs. While sharing energy is the main focus of this work, our proposed method applies to a vast spectrum of application domains beyond energy management, such as computational task offloading, distributed learning, and the like.
\subsection{Related Work}
Optimization strategies are well-known solutions to optimal DERs management problems. Our work adds an essential factor to the state-of-the-art methods, namely, uncertainty and lack of information. As such, the research related to ours stems from two main categories, namely, distributed energy resources and online convex optimization under bandit feedback. 

Reference \cite{guo2012optimal} investigates optimal power management for residential customers in a smart grid (SG) combined with renewable energy generation and battery storage. The proposed solution is based on Lyapunov optimization. It can achieve close-to-optimal performance with a trade-off between battery capacity and cost savings. Similarly, \cite{salinas2013dynamic} adopts a Lyapunov optimization method to solve the energy management problem for SG under unpredictable load demands of users. Specifically, users in the energy distribution network have renewable energy resources, an energy storage device, and a connection to the power grid, which collaboratively satisfy their load demands. The proposed dynamic energy management scheme optimally schedules the usage of all energy resources based on the current system state only. However, the strategies above are all central and require a controller/coordinator. Reference investigates the energy management problem by formulating it as a convex optimization problem, and develops a centralized offline solution. Reference \cite{wang2014distributed} extends \cite{wang2013online} by designing distributed algorithms to solve the formulated problem. Similarly,  in \cite{wu2016distributed}, the authors consider an optimal DERs coordination problem over multiple time periods subject to constraints and solve it via a distributed consensus algorithm based on gradient strategy. Reference \cite{salazar2020energy} develops a simplified Markov model for photovoltaic power generation and proposes a stochastic dynamic programming optimization framework to solve the energy management problem. To reduce the communication cost of distributed structure in energy management, reference \cite{ding2018distributed} designs an optimization algorithm with event-triggered communication and control mechanism.

Besides, most strategies require critical information, such as the consumers' demands and usage, which scarifies the users' privacy. Consequently, several authors propose privacy-preserving methods for distributed energy management \cite{nizami2019multiagent,khorasany2020framework,ye2021scalable}.

In addition to the solution based on optimization theory, some prior works study the energy management in DERs from a game-theoretical \cite{mohsenian2010autonomous,nguyen2015decentralized,wang2019incentive,8585045} or a reinforcement learning perspective \cite{ye2021scalable,wan2021price}. In our work, we propose a distributed online convex optimization strategy with privacy preservation. Therefore, we do not elaborate game-theoretical methods intensively. 

Zinkevich \textit{et al.} \cite{zinkevich2003online} design an online convex optimization (OCO) algorithm based on gradient descent inspired by the infinitesimal gradient ascent concept from repeated games. The regret of the learning method is upper bounded by $\mathcal{O}(\sqrt{T})$. Reference \cite{Flaxman2004BCO} extends \cite{zinkevich2003online} to a bandit setting, namely, \textbf{bandit convex optimization (BCO)}, where in each period, only the cost-utility ratio is observable. The proposed gradient-free method obtains a regret bound $\mathcal{O}(T^{3/4})$ for the general case for bounded and Lipschitz-continuous convex loss function. The authors of \cite{hazan2014bandit} propose a near-optimal algorithm for BCO with strongly-convex and smooth loss functions and prove a regret bound of $\tilde{\mathcal{O}}(\sqrt{T})$. Besides, they introduce a self-concordant barrier function to solve constrained BCO problems. In that research, the authors assume that the constraint always holds, and apply a projection step. Nevertheless, considering the constraints at the projection step is computationally expensive \cite{jenatton2016adaptive}. Besides, in practical applications, the learner may be concerned with long-term constraints. To alleviate those difficulties, OCO and BCO under constraints are potential solutions. Reference \cite{mahdavi2012trading} considers the online convex optimization problem with long-term constraints. The proposed algorithm achieves $\mathcal{O}(\sqrt{T})$ regret bound and $\mathcal{O}(T^{3/4})$ bound of constraint violation. Following \cite{mahdavi2012trading}, \cite{jenatton2016adaptive} proposes an adaptive online gradient descent algorithm to solve online convex optimization problems with long-term constraints. Reference \cite{yu2017online} considers the online convex optimization problem with stochastic constraints, which can achieve $\mathcal{O}(\sqrt{T})$ regret bound and constraint violations. It further proves $\mathcal{O}(\sqrt{T} \log T)$ high probability regret bound and constraint violations.

Finally, distributed convex optimization \cite{li2020online,liang2019distributed,Yi2020distr_band_online_conv_opt} and convex optimization in non-stationary environment \cite{zhao2021bandit} attract increasing attention. In this paper, we solve the energy management problem in DERs by formulating it as a distributed BCO problem with constraints in a non-stationary environment. 

\subsection{Contributions}
In this paper, we investigate the real-time energy management problem for networked DERs systems consisting of a number of nodes, each of which is an energy- generator and consumer simultaneously. It is an extension of our previous work \cite{tsetis2023bandit}, however, our previous work only considers the distributed energy management in a stationary environment while in this work we consider a more complicated problem. We list our contributions briefly below.
\begin{itemize}
\item We investigate an energy-sharing problem in a network of systems. We minimize the detrimental impact of energy shortage on systems' performance because of uncertainty in energy production or harvesting. We model this problem as an online convex optimization in a non-stationary environment.
\item In our setting, we consider a bandit feedback, where each node/user only observes the energy satisfaction level of its neighbors. Thus, the nodes in the network can maintain privacy and do not reveal their consumption and generation level to other nodes.
\item To solve the proposed energy resource sharing problem with a gradient-free method, we develop an algorithm, namely, Distributed Resource Sharing with regularized Lagrangian function, inspired by \cite{Yi2020distr_band_online_conv_opt} and its extension in a dynamic environment, which achieve a dynamic regret bound of $\tilde{\mathcal{O}}((1+P_T)^{\frac{1}{2}}T^{\frac{1}{2}} + (1+P_T)^{\frac{1}{4}}T^{\frac{3}{4}})$ and $\tilde{\mathcal{O}}(T^{\frac{3}{4}}(1+P_T)^{\frac{1}{2}})$, respectively, where $T$ refers to the total number of rounds and $P_T$ refers to the path-length of comparators. Table~\ref{tab:stoa} compares our proposed algorithm with the state-of-the-art. 
\item To reduce the constraint violation in real-world implementation, we propose one adjustment step that does not worsen the regret bound. 
\item We evaluate our proposed algorithm using a real-world dataset and compare it with state-of-the-art algorithms. Experimental results show that our proposal reduces energy waste, thereby preventing an environmental damage. 
\end{itemize}
\begin{table*}[!ht]
\centering
\captionsetup{justification=centering}
\caption{Our two proposed algorithms compared to the SOTA research in BCO.}
\label{tab:stoa}
 \begin{tabular}{c|c|c|c|c|c|c}
 \hline
 Reference & Problem Type & Constraint & Feedback & Hyper parameter & Dynamic Regret & Meta-learning \\ [0.5ex] 
 \hline
 \cite{mahdavi2012trading} & Centralized & \XSolidBrush & Bandit & Constant & \XSolidBrush & \XSolidBrush \\
 \hline
 \cite{Yi2020distr_band_online_conv_opt} & Distributed & \Checkmark & Bandit & Time-varying & \Checkmark & \XSolidBrush \\
 \hline
 \cite{zhao2021bandit} & Centralized & \XSolidBrush & Bandit & Constant & \Checkmark & \Checkmark \\
 \hline
 \cite{guo2022online} & Centralized & \Checkmark & Full Information & Time-varying & \Checkmark & \Checkmark \\
 \hline
 DRS & Distributed & \Checkmark & Bandit & Time-varying & \Checkmark & \XSolidBrush \\
 \hline
 MA-NSDRS & Distributed & \Checkmark & Bandit & Time-varying & \Checkmark & \Checkmark \\
 \hline 
 \end{tabular}
\end{table*}
The rest of this paper is organized as follows. In Section~\ref{sec:problem}, we formulate the problem.  Section~\ref{sec:solution} motivates and explains the ``Distributed Resource Sharing'' (DRS) algorithm and its improved version ``Meta Algorithm for Non-stationary Distributed Resource Sharing'' (MA-NSDRS). Section~\ref{sec:solution} includes the theoretical analysis. Section~\ref{sec:experiment} demonstrates the numerical results and Section~\ref{sec:conclusion} summarizes our work. All proofs appear in the Appendix. In Section~\ref{sec:experiment}, we demonstrate the numerical simulations using a real-world dataset from the DERs data of New York State \cite{NY_DERs}. Finally, Section~\ref{sec:conclusion} concludes the work and suggests some future research directions. 

\section{Problem Formulation}
\label{sec:problem}
Consider a smart power system with multiple nodes collected in the set $\mathcal{N} = \{1,2,\ldots,N\}$. Each node has a renewable energy generator and an energy consumption module. Hence, it is simultaneously an energy producer and a consumer, i.e., a prosumer \cite{khorasany2020framework,ye2021scalable}. Each node might have distinct energy production and energy consumption levels. Thus, it is connected to some other nodes to as a potential compensation for his energy shortage, or to donate energy. Specially, if the energy transmission loss between two nodes is expected to exceed a specified threshold fraction of the transmitted energy, the nodes do not exchange energy, resulting in no edge between them in the graph. Conversely, an edge exists if the transmission loss is below this threshold. Thus, our model considers DERs in a small graphical area with little energy transmission loss. In addition to the self-generated energy, each node also connects to the main power grid. Note that the main power grid receives input from the traditional power generation plants, which incurs high power loss during the transmission and relies mainly on fossil fuel consumption. As such, it is not the ideal source of the energy. Each node firstly relies on the local energy generation and the received energy from its neighbors in a small graphical area. The main power grid only assists in vital cases as the last choice, and we consider the simplified model with the main grid with unlimited capacity. Beyond that, it does not play any role in our setting; Therefore, we do not include it in our mathematical model. Figure~\ref{fig:sys-ders} depicts the system model .
\begin{figure}[!ht]
\centering
\includegraphics[width=0.5\linewidth]{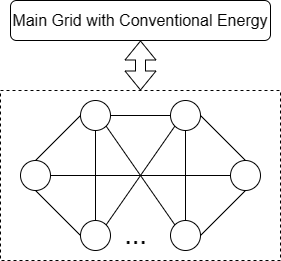}
\caption{System Model.}
\label{fig:sys-ders}
\end{figure}

The DERs transmission follows the undirected network graph $\mathcal{G}=(\mathcal{N},\mathcal{E})$, i.e., alongside the edges $\mathcal{E}$. Let $\mathcal{N}_i=\{j: (i,j) \in \mathcal{E}\}$ be the set of node $i$'s neighbors and $\tilde{\mathcal{N}}_i=\mathcal{N}_i \cup \{i\}$, the set that includes node $i$'s neighbors and node $i$. Similarly, $\textbf{A}$ is the adjacency matrix of the network, where $A_{i,j} = 1$ if $(i,j) \in \mathcal{E}$ and $\tilde{\textbf{A}}=\textbf{A}+\textbf{I}$ when the neighboring nodes include the node itself. 

At each round $t$, node $i$'s energy generation amount is $d_{i,t}$, which is a random variable with expectation $\mu_{i,t}$. For simplicity, we assume the generation of each node is independent of each other in the network. Besides, the energy demand of node $i$ at time step $t$ is $l_{i,t}$. It can keep its generated resources or transfer $x_{i,t}(k)$ to its neighbor $k \in \mathcal{N}_i$. The allocation vector of node $i$ yield $\boldsymbol{x}_{i,t} = [x_{i,t}(1),...,x_{i,t}(N)] \in \mathcal{X}_{i}$, where $\mathcal{X}_{i} \subseteq \mathbb{R}_+^{N}$. Finally, at each round, the total resources of agent $i$ yields $\sum_{k \in \tilde{\mathcal{N}}_i} x_{k,t}(i)$, which includes the received resources $\sum_{k \in \mathcal{N}_i} x_{k,t}(i)$ and the amount node $i$ reserves for itself $x_{i,t}(i)$.

We define node $i$'s utility in terms of its energy satisfaction level compared to its demand $l_{i,t}$. Formally,
\begin{gather}
\bar{f}_{i,t} = \min\{\frac{\sum_{k \in \tilde{\mathcal{N}}_i} x_{k,t}(i)}{l_{i,t}} , 1\}.
\end{gather}
Besides, distributing energy has a hard (feasibility) constraint: The energy that  node $i$ distributes to its neighbors should not exceed the generation. Formally,
\begin{gather}
g_{i,t} = \sum_{j \in \tilde{\mathcal{N}}_i} x_{i,t}(j) - d_{i,t} = \langle \boldsymbol{\tilde{a}}_{i}, \boldsymbol{x_{i,t}} \rangle - d_{i,t} \leq 0, 
\label{eq:constr-g}
\end{gather}
where $\boldsymbol{\tilde{a}}_{i}$ is the $i$-th row of the matrix $\boldsymbol{\tilde{A}}$. Considering the cooperation among agents, we define the loss function as 
\begin{gather}
f_{i,t} = 1-\frac{1}{\vert \tilde{\mathcal{N}}_i \vert} \sum_{j \in \tilde{\mathcal{N}}_i}\bar{f}_{j,t} = 1-\frac{1}{\vert \tilde{\mathcal{N}}_i \vert} \langle \boldsymbol{\tilde{a}}_i, \boldsymbol{\bar{f}}_t \rangle, 
\label{eq:loss_f}
\end{gather}
where $\boldsymbol{\bar{f}}_t = [\bar{f}_{1,t}, \ldots, \bar{f}_{N,t}]$. In other words, the loss function measures the shortage level of each node's neighbors after the cooperative energy allocation. For simplicity, we assume that the neighbor nodes are close enough so that the energy transfer cost is negligible. However, extending our methodology to the case of costly energy transfer is straightforward. Below, we state assumptions concerning the loss function \eqref{eq:loss_f} and the constraint function \eqref{eq:constr-g} and a proposition concerning the loss function. 
\begin{assumption}\label{assump:fg-bound}
The loss function $f$ and the constraint function are bounded. We denote their bounds by $F$ and $G$, respectively.
\end{assumption}
\begin{assumption}\label{assump:x}
For each agent, the action space $\mathcal{X}_i$ is inside a ball with radius $R_i$ and contains a ball with radius $r_i$.
\end{assumption}
\begin{proposition} 
\label{prop:convex-loss}
The loss function (\ref{eq:loss_f}) is convex and Lipschitz continuous with the lipschitz constant $L = \frac{1}{\min_{j \in \mathcal{N}} l_j}$. 
\end{proposition}
\begin{proof}
See Appendix~\ref{app:convex-loss}.
\end{proof}
\begin{remark}
The proposed loss function \eqref{eq:loss_f} measures the energy satisfaction level in each node's neighborhood, which is proportional to the extra energy charged from the main grid. Under this setup, we aim to minimize the total energy cost of the conventional energy drawn from the main grid. 
\end{remark}
\begin{remark}
We can introduce a variation of the loss function \eqref{eq:loss_f} by incorporating a discount factor $c(i,j) \in (0,1]$ for edge $(i,j)$ in the network. This factor models the energy dissipation through transfer. For this modified condition, we define the loss function as
\begin{align}
    f_{i,t} &= 1-\frac{1}{\vert \tilde{\mathcal{N}}_i \vert} \sum_{j \in \tilde{\mathcal{N}}_i}\bar{f}_{j,t} \notag \\
    &= 1-\frac{1}{\vert \tilde{\mathcal{N}}_i \vert} \sum_{j \in \tilde{\mathcal{N}}_i}\min\{\frac{\sum_{k \in \tilde{\mathcal{N}}_i} c(i,k)x_{k,t}(i)}{l_{i,t}} , 1\}.
\end{align}
Note that this adjustment does not compromise the performance of our algorithm, as it is agnostic to the specific form of any convex and Lipschitz continuous loss function. As such, our proposed algorithm remains applicable and effective. In general, a class of loss function that satisfy Assumption~\ref{assump:fg-bound} and \ref{assump:x} is workable under our proposed algorithm. For simplicity and consistency, we continue to utilize the loss function defined in \eqref{eq:loss_f}. 
\end{remark}
\subsection{Problem Formulation}
The energy management procedure follows in successive rounds $t = 1,\ldots,T$. It is widely recognized that the generation of renewable energy resources has random and intermittent characteristics inherently \cite{Maghsudi17:DUA}; Hence, we model it as a non-stationary process. Consequently, the energy shortage at each node is also non-stationary. 

Our objective is to develop an algorithm for each node to minimize the cumulative loss, given by \eqref{eq:loss_f}, with minimized constraint violations, given by \eqref{eq:constr-g}. That is equivalent to regret minimization in a bandit setting. We assume that nodes fully comply with the output of the algorithm. For a fixed time horizon $T$, the static regret is the cost of not playing the optimal action in the hindsight. Formally,
\begin{gather}
S-R_T^i = \mathbb{E}[\sum_{t=1}^T f_{i,t}(\boldsymbol{x}_{i,t})] - \min_{\boldsymbol{x} \in \mathcal{X}}\sum_{t=1}^T f_{i,t}(\boldsymbol{x}).  
\label{eq:s-rgt}
\end{gather}
The static regret implicitly assumes that there is a reasonably good decision over all iterations \cite{zhao2021bandit}. However, in non-stationary environments, the underlying distribution of the online reward function changes \cite{sugiyama2012machine,gama2014survey,zhao2019distribution}. To address this limitation, we define the dynamic regret as the difference between the cumulative loss of the player $i$ and that of a comparator sequence $\boldsymbol{u}_{i,t} \in \mathcal{X}$. Formally, 
\begin{gather}
    D-R_T^i(\boldsymbol{u}_{i,1}, \ldots,\boldsymbol{u}_{i,T}) = \mathbb{E}[\sum_{t=1}^T f_{i,t}(\boldsymbol{x}_{i,t})] - \sum_{t=1}^T f_{i,t}(\boldsymbol{u}_{i,t}).  
\label{eq:d-rgt}
\end{gather}
In contrast to the static regret, the dynamic regret \eqref{eq:d-rgt} compares with a sequence of changing comparators; Hence, it is more suitable to measure the algorithms' performance under a non-stationary environment. 

Besides regret, constraint violation, defined below, is a performance metric of the algorithm. 
\begin{gather}
V_T^i = \sum_{t=1}^T g_{i,t}^+(\boldsymbol{x}_{i,t}),
\end{gather}
where the operator $(\cdot)^+ = \max (\cdot, 0)$ adds all the violations over all iterations. We call it hard constraint and the metric cumulative hard constraint violation \cite{guo2022online}. 
\section{Proposed Solution}
\label{sec:solution}
We first propose an algorithm to solve the formulated problem in a stationary environment. Afterward, we generalize the algorithm to deal with the non-stationary environment. 
\subsection{Distributed Resource-sharing Algorithm}
We summarize our developed distributed resource-sharing algorithm (DRS) with bandit feedback in \textbf{Algorithm~\ref{alg:drs}}. Detailed description follows.
\begin{algorithm}[!ht]
\caption{DRS for $i \in \mathcal{N}$}\label{alg:drs}
\begin{algorithmic}[1]
\STATE \textbf{Initialize:} 
$\boldsymbol{u}_{i,1} \in \mathcal{S}^{\tilde{\mathcal{N}}_i}$, $\boldsymbol{z}_{i,1} \in (1-\xi_{i,t})\mathcal{X}_i$, $\boldsymbol{x}_{i,1}=\boldsymbol{z}_{i,1}+\delta_{i,1}\boldsymbol{u}_{i,1}$, $\boldsymbol{x}_{i,1}'= \boldsymbol{x}_{i,1}$, $i \in [n]$.
\FOR{$t=2,...,T$} 
    \STATE Play action $\boldsymbol{x}'_{i,t-1}$.
    \STATE Get $\bar{f}_{i,t-1}(\boldsymbol{x}'_{i,t-1})$ and $g_{i,t-1}(\boldsymbol{x}'_{i,t-1})$.
    \STATE Communicate with neighbors and obtain $f_{i,t-1}$ according to (\ref{eq:loss_f}).
    \STATE Select random unit vector $\boldsymbol{u}_{i,t}$. Let 
    \begin{align}
        \boldsymbol{z}_{i,t} &=\mathcal{P}_{(1-\xi_{i,t})\mathcal{X}_i} (\boldsymbol{z}_{i,t-1}-\eta_{i,t}\hat{\boldsymbol{\Gamma}}_{i,t}) 
        \label{eq:proj-z}\\  \boldsymbol{x}_{i,t}&=\boldsymbol{z}_{i,t}+\delta_{i,t}\boldsymbol{u}_{i,t} 
        \label{eq:x} \\
        q_{i,t} &= (q_{i,t-1}+\gamma_{i,t}(g_{i,t-1}(\boldsymbol{x}_{i,t-1}')-\beta_{i,t}q_{i,t-1}))_{+} \label{eq:proj-q}
    \end{align}
    \IF{$\sum_{j \in \tilde{\mathcal{N}}_i} x_{i,t}(j) > d_{i,t}$}
    \STATE 
    \begin{gather}
        x'_{i,t}(i)= \frac{x_{i,t}(i)}{\sum_{j=1}^{N} x_{i,t}(j)} d_{i,t}, \label{eq:sio}
    \end{gather}
    \ELSE 
    \STATE $\boldsymbol{x}_{i,t}' = \boldsymbol{x}_{i,t}$.
    \ENDIF
\ENDFOR
\end{algorithmic}
\end{algorithm}  

In DRS, the agents run the online bandit convex optimization strategy in parallel. Let $\boldsymbol{z}$ be an updating variable. The action vector $\boldsymbol{x}$ denotes an estimation of $\boldsymbol{z}$, i.e., $\boldsymbol{x} = \boldsymbol{z} + \delta \boldsymbol{u}$, where $\boldsymbol{u}$ is a randomly generated unit vector. The DRS algorithm firstly initializes the parameter by randomly sampling $\boldsymbol{u}_{i,1} \in \mathcal{S}^{\tilde{\mathcal{N}}_i}$ and $\boldsymbol{z}_{i,1} \in (1-\xi_{i,t})\mathcal{X}_i$. It then calculates the corresponding $\boldsymbol{x}_{i,1}$ and $\boldsymbol{x}'_{i,1}$.

At every round $t$, each agent $i$ plays an action $\boldsymbol{x}_{i,t}'$ that yields some individual loss $\bar{f}_{i,t-1}$. After receiving feedback from its neighbors, each agent obtains the final loss $f_{i,t-1}$ based on its own and neighbors' losses. Subsequently, the algorithm updates the estimations using the bandit gradient descent method via two projection functions. Motivated by \textit{regularized Lagrangian} proposed in \cite{mahdavi2012trading,jenatton2016adaptive,Yi2020distr_band_online_conv_opt}, we define the regularized Lagrangian function for DRS as
\begin{gather}
\mathcal{L}_{i,t}(\boldsymbol{x},q) = f_{i,t}(\boldsymbol{x}) + q g_{i,t}(\boldsymbol{x}) - \frac{\beta}{2} q^2, \label{eq:lagr}
\end{gather}
where $f_{i,t}$ and $g_{i,t}$ respectively denote the loss- and constraint functions, while $q \in \mathbb{R}_{+}$ is the dual variable. The update steps (\ref{eq:proj-z}) and (\ref{eq:proj-q}) are originally from
\begin{gather}
\boldsymbol{x}_{i,t+1} = \mathcal{P}_{\mathcal{X}}(\boldsymbol{x}_{i,t} - \eta_i \nabla_{\boldsymbol{x}} \mathcal{L}_{i,t}(\boldsymbol{x}_{i,t},q_{i,t})), 
\label{eq:proj_x}\\
q_{i,t+1} = (q_{i,t}+\gamma_i\nabla_{q} \mathcal{L}_{i,t}(\boldsymbol{x}_{i,t},q_{i,t}))_+,
\label{eq:p_q}
\end{gather}
where $\nabla_{\boldsymbol{x}} \mathcal{L}_{i,t}(\boldsymbol{x}_{i,t},q_{i,t})) = \nabla f_{i,t}(\boldsymbol{x}_{i,t})+(\nabla g_{i,t}(\boldsymbol{x}_{i,t}))^Tq_{i,t}$ and $\nabla_{q} \mathcal{L}_{i,t}(\boldsymbol{x}_{i,t},q_{i,t})) = g_{i,t}(\boldsymbol{x}_{i,t}) - \beta q_{i,t}$. Lemma~\ref{lem:opeg} offers the approximation of the gradient under bandit feedback.
\begin{lemma}[\cite{Flaxman2004BCO}]
\label{lem:opeg}
Fix $\delta > 0$ as a small smoothing parameter. Over the random unit vectors $\boldsymbol{u}$,
\begin{gather}
\mathbb{E}_{\boldsymbol{u}}[f(\boldsymbol{x}+\delta\boldsymbol{u})] = \hat{f}(\boldsymbol{x}) \\
\mathbb{E}_{\boldsymbol{u}}[f(\boldsymbol{x}+\delta \boldsymbol{u})\boldsymbol{u}] = \frac{\delta}{N}\hat{\nabla}f(\boldsymbol{x}), \label{eq:est-grad}
\end{gather}
where $N$ is the dimension of $\boldsymbol{x}$.
\end{lemma}
By Lemma~\ref{lem:opeg}, the updated direction information under bandit feedback yields
\begin{align}
\hat{\boldsymbol{\Gamma}}_{i,t} &=\hat{\nabla}f_{i,t-1}(\boldsymbol{z}_{i,t-1})+(\hat{\nabla}g_{i,t-1}(\boldsymbol{z}_{i,t-1}))^Tq_{i,t} \label{eq:a} \\
&=\frac{N}{\delta_{i,t-1}}[ f_{i,t}(\boldsymbol{x}_{i,t-1})\boldsymbol{u}_{i,t-1}+(g_{i,t-1}(\boldsymbol{x}_{i,t-1})\boldsymbol{u}_{i,t-1})^T q_{i,t}]. \label{eq:est-grad} 
\end{align}
Besides, (\ref{eq:proj-q}) is the same as projection function (\ref{eq:p_q}) without any approximations. 

Most cutting-edge methods to handle time-varying constraints prevent constraint violations in the long run, e.g., on average, while ignoring intermediary rounds. However, in reality, such a setting can be severely harmful and even make the problem infeasible. To deal with that issue, we propose adjustments in our method that reduce or eliminate constraint violations. If the aggregated energy distribution of a node exceeds its existing generated resources, then it calculates a new distribution profile to its neighborhood (including itself) $\boldsymbol{x}_i \in \mathbb{R}^{\tilde{\mathcal{N}}_i}$ according to \eqref{eq:sio}. 
\subsubsection{Performance Guarantees}
The following theorem states the regret bound of our proposed algorithm.
\begin{theorem}
\label{the:drs-drgt}
Set $\delta_{i,t} = (\frac{|\mathcal{N}_i|F_i}{\tilde{L}})^{\frac{1}{2}}(\frac{R^2/2+RP_T}{t})^{\frac{1}{4}}$, $\eta_{i,t} = (\frac{1}{|\mathcal{N}_i|F_i\tilde{L}})^{\frac{1}{2}}(\frac{R^2/2+RP_T}{t})^{\frac{3}{4}}$, $\beta_{i,t} = \frac{1}{G_it^{1/2}}$, $\gamma_{i,t} = \frac{1}{G_i^2t^{1/2}}$ and $\xi_{i,t} = \delta_{i,t}/r_i$ with $\tilde{L} = 3L + LR/r$ and Assumption~\ref{assump:fg-bound} and \ref{assump:x} hold. The expected dynamic regret and the violation of constraints of \textbf{Algorithm~\ref{alg:drs}} satisfy
\begin{gather}
D-R_T^i \leq \tilde{\mathcal{O}}((1+P_T)^{\frac{1}{2}}T^{\frac{1}{2}} + (1+P_T)^{\frac{1}{4}}T^{\frac{3}{4}}) \\
V_T^i \leq \tilde{\mathcal{O}}([(1+P_T)^{\frac{1}{2}}T+(1+P_T)^{\frac{1}{4}}T^{\frac{5}{4}}]^{\frac{1}{2}}),
\end{gather}
for any comparator sequence $\boldsymbol{u}_1, \ldots, \boldsymbol{u}_T \in \mathcal{X}$ and $\tilde{\mathcal{O}}(\cdot)$ omits $\log T$ factors. In above, $P_T = \sum_{t=2}^T\norm{\boldsymbol{u}_t - \boldsymbol{u}_{t-1}}$ is the path-length.
\end{theorem}
\begin{proof}
See appendix~\ref{app:drs-drgt}.
\end{proof}
The following corollary characterises the regret of Algorithm~\ref{alg:drs} for the path length $P_T = 0$.
\begin{corollary}
\label{cor:drs-srgt}
%
%
Set the hyper parameters $\delta_{i,t} = (\frac{|\mathcal{N}_i|F_i}{\tilde{L}})^{\frac{1}{2}}(\frac{R^2/2}{t})^{\frac{1}{4}}$, $\eta_{i,t} = (\frac{1}{|\mathcal{N}_i|F_i\tilde{L}})^{\frac{1}{2}}(\frac{R^2/2}{t})^{\frac{3}{4}}$, $\beta_{i,t} = \frac{1}{G_it^{1/2}}$, $\gamma_{i,t} = \frac{1}{G_i^2t^{1/2}}$s and $\xi_{i,t} = \delta_{i,t}/r_i$. If Assumptions~\ref{assump:fg-bound} and \ref{assump:x} hold, the static regret of \textbf{Algorithm~\ref{alg:drs}} satisfies
\begin{gather}
S-R_T^i \leq \tilde{\mathcal{O}}(T^{3/4}), \quad V_T^i \leq \tilde{\mathcal{O}}(T^{\frac{5}{8}}).
\end{gather}
\end{corollary}
\begin{corollary}
\label{cor:drs-adj}
Set the hyperparameters $\delta_{i,t} = (\frac{|\mathcal{N}_i|F_i}{\tilde{L}})^{\frac{1}{2}}(\frac{R^2/2+RP_T}{t})^{\frac{1}{4}}$, $\eta_{i,t} = (\frac{1}{|\mathcal{N}_i|F_i\tilde{L}})^{\frac{1}{2}}(\frac{R^2/2+RP_T}{t})^{\frac{3}{4}}$, $\beta_{i,t} = \frac{1}{G_it^{1/2}}$, $\gamma_{i,t} = \frac{1}{G_i^2t^{1/2}}$ and $\xi_{i,t} = \delta_{i,t}/r_i$. If Assumptions~\ref{assump:fg-bound} and \ref{assump:x} hold, with our adjustment steps (Line 7-10) in Algorithm \ref{alg:drs}, the dynamic regret bound and static regret bound have the same growth order as 
\begin{gather}
D-R_T^i \leq  \tilde{\mathcal{O}}((1+P_T)^{\frac{1}{2}}T^{\frac{1}{2}} + (1+P_T)^{\frac{1}{4}}T^{\frac{3}{4}}) \\
S-R_T^i \leq \tilde{\mathcal{O}}(T^{3/4}). 
\end{gather}    
\end{corollary}
\begin{proof}
See appendix~\ref{app:drs-adj}.
\end{proof}
\subsection{Meta Algorithm for Non-stationary DRS}
As stated by Theorem~\ref{the:drs-drgt}, selecting the hyperparameters that minimize the regret depends on $P_T$; Nevertheless, $P_T$ is unknown and difficult to predict in real-world applications. The online ensemble is a solution to address that problem. It maintains candidates in parallel and uses expert-tracking algorithms to combine predictions and track the best parameter in a grid search manner \cite{van2016metagrad}. 

Nonetheless, in bandit convex optimization, implementing the online ensemble is cumbersome since only the reward/loss function value is available. The parameter-free algorithm proposed in \cite{zhao2021bandit} mitigates this limitation in the bandit convex optimization without considering the constraint violation. Motivated by \cite{zhao2021bandit} and \cite{guo2022online}, we propose our MA-NSDRS (Meta algorithm for non-stationary distributed resource sharing) algorithm, which is also an online ensemble method. More precisely, it combines a meta algorithm and several expert algorithms with our proposed DRS algorithm (Algorithm~\ref{alg:drs}). Details follow. 
\paragraph{Expert Algorithm} Algorithm~\ref{alg:ma-nsdrs} maintains the pool of candidate step sizes $\mathcal{H}$ for each agent and initializes an expert for each candidate step size. 

According to Theorem~\ref{the:drs-drgt}, the best possible step size is $\eta_{i,t}^{\dagger} = \sqrt{R_i(2R_i^2+R_iP_T)/(|\mathcal{N}_i|F_i\tilde{L})}\cdot t^{-3/4}$. Due to the unknown path-length $P_T$, this step size is not available. Since $ 0 \leq P_T \leq 2R_iT$, one can conclude that the best step size lies in the range
\begin{gather}
\sqrt{\frac{2R_i^3}{|\mathcal{N}_i|F_i\tilde{L}}}t^{-\frac{3}{4}} \leq \eta^{\dagger} \leq  \sqrt{\frac{2R_i^3(1+T)}{|\mathcal{N}_i|F_i\tilde{L}}}t^{-\frac{3}{4}}.
\end{gather}
Therefore, the pool of candidate step sizes $\mathcal{H}$ yields 
\begin{gather}
\mathcal{H} = \Big\{\eta_k = 2^{k-1}\sqrt{\frac{2R^3}{2|\mathcal{N}_i|F_i\tilde{L}}}\cdot t^{-\frac{3}{4}}| k= 1,2,\ldots, K \Big \}, 
\label{eq:pool}
\end{gather}
where $K = \lceil \frac{1}{2}\log_2(1+T)\rceil + 1$ is the number of candidate step sizes. The configuration \eqref{eq:pool} ensures the existence of an index $k^* \in [K]$ such that $\eta_{k^*} \leq \eta^{\dagger} \leq 2\eta_{k^*}$; i.e., there exists one step size in $\mathcal{H}$ which, despite being suboptimal, is sufficiently close to the optimal step size $\eta^{\dagger}$.     

Define the surrogate loss $l_t:(1-\xi)\mathcal{X} \rightarrow \mathbb{R}$ as
\begin{gather}
    l_t(\boldsymbol{x}) = \langle \hat{\nabla}f_{i,t}(\boldsymbol{x}), \boldsymbol{x} - \boldsymbol{x}_t \rangle, 
\end{gather}
where $\hat{\nabla}f_{i,t}(\boldsymbol{x})$ is the estimated gradient. The properties of the surrogate loss are as follow \cite{zhao2021bandit}.
\begin{property}
\label{prop:sur-loss1}
For any $\boldsymbol{x} \in (1-\xi)\mathcal{X}$, $\nabla l_t(\boldsymbol{x}) = \hat{\nabla}f_{i,t}(\boldsymbol{x})$. 
\end{property}
\begin{property}
For any $\boldsymbol{v} \in (1-\xi)\mathcal{X}$, $\mathbb{E}[\hat{f}_t(\boldsymbol{x}_t) - \hat{f}_t(\boldsymbol{v})] \leq \mathbb{E}[l_t(\boldsymbol{x}_t) -l_t(\boldsymbol{v})]$
\end{property}
At time step $t$, each expert $k$ runs the online gradient descent
\begin{align}
    \boldsymbol{z}_{i,t}^k &=\mathcal{P}_{(1-\xi_{i,t})\mathcal{X}_i} (\boldsymbol{z}_{i,t-1}^j-\eta_k\hat{\boldsymbol{\Gamma}}_{i,t}), \label{eq:proj-ez}
\end{align}
where $\eta_k \in \mathcal{H}$ is the step size of the expert $k$, $\hat{\boldsymbol{\Gamma}}_{i,t}$ is defined in \eqref{eq:est-grad}. Thanks to the Property~\ref{prop:sur-loss1}, all experts can perform the exact online gradient descent in the same direction $\hat{\boldsymbol{\Gamma}}_{i,t}$ at each time step $t$ \cite{zhao2021bandit}. 

\paragraph{Meta Algorithm} To combine the predictions from multiple experts, we utilize the exponentially weighted average forecaster algorithm \cite{cesa2006prediction} with non-uniform initial weights as meta algorithm (Line 16). The modified meta algorithm, together with our proposed DRS algorithm, yields our proposed MA-NSDRS algorithm for minimizing the dynamic regret of our formulated problem in non-stationary environments. The MA-NSDRS algorithm is summarized in Algorithm~\ref{alg:ma-nsdrs}.

\begin{algorithm}[!ht]
\caption{MA-NSDRS: Meta algorithm for non-stationary distributed resource sharing}\label{alg:ma-nsdrs}
\begin{algorithmic}[1]
\STATE \textbf{Input:} time horizon $T$, number of experts $K$, pool of candidate step sizes $\mathcal{H} = \{\eta_1,\ldots,\eta_K\}$, learning rate of the meta algorithm $\epsilon$.
\STATE \textbf{Initialize:} 
\begin{itemize}
    \item $\boldsymbol{u}_{i,0} \in \mathcal{S}^{\tilde{\mathcal{N}}_i}$, $\boldsymbol{z}_{i,0} \in (1-\xi_i)\mathcal{X}_i$, $\boldsymbol{x}_{i,0}=\boldsymbol{z}_{i,0}+\delta_i\boldsymbol{u}_{i,0}$, $\boldsymbol{x}_{i,0}'= \boldsymbol{x}_{i,0}$, $i \in [N]$.
    \item (K experts) $\omega_{i,1}^k = \frac{K+1}{K}\frac{1}{k(k+1)}$
\end{itemize} 
\FOR{$t=1,\ldots,T$} 
\STATE Receive $\boldsymbol{x}_{i,t}^k$ from each expert $k$
\begin{align}
    \boldsymbol{z}_{i,t}^k &=\mathcal{P}_{(1-\xi_i)\mathcal{X}_i} (\boldsymbol{z}_{i,t-1}^j-\eta_k\hat{\boldsymbol{\Gamma}}_{i,t}) \label{eq:proj-ez} 
    \end{align}
\STATE Obtain 
$\boldsymbol{z}_{i,t} = \sum_{k \in [K]} \omega_{i,t}^k \boldsymbol{z}_{i,t}^k$
\STATE Obtain $\boldsymbol{x}_{i,t}=\boldsymbol{z}_{i,t}+\delta_i\boldsymbol{u}_{i,t}$ 
\IF{$\sum_{j \in \tilde{\mathcal{N}}_i} x_{i,t}(j) > d_{i,t}$}
    \STATE 
    \begin{gather}
        x'_{i,t}(i)= \frac{x_{i,t}(i)}{\sum_{j=1}^{N} x_{i,t}(j)} d_{i,t}, 
    \end{gather}
\ELSE
    \STATE $\boldsymbol{x}_{i,t}' = \boldsymbol{x}_{i,t}$.
\ENDIF
\STATE Update 
\begin{gather}
    q_{i,t} = (q_{i,t-1}+\gamma_i(g_{i,t-1}(\boldsymbol{x}_{i,t-1}')-\beta_iq_{i,t-1}))_{+} 
\end{gather}
\STATE Play action $\boldsymbol{x}_{i,t-1}'$ and receive $\bar{f}_{i,t-1}(\boldsymbol{x}'_{i,t-1})$ and $g_{i,t-1}(\boldsymbol{x}'_{i,t-1})$.
\STATE Communicate with neighbors and obtain $f_{i,t-1}$ according to (\ref{eq:loss_f}).
\STATE Compute gradient and construct surrogate loss as 
\begin{gather}
    l_{i,t}(\boldsymbol{z}) = \langle \hat{\nabla}f_{i,t}, \boldsymbol{z}- \boldsymbol{z}_{i,t} \rangle.  \label{eq:sur-loss}
\end{gather}
\STATE Update the weight of each expert $k$ by 
\begin{gather}
    \omega_{i,t+1}^k = \frac{\omega_{i,t}^k \exp(-\epsilon l_{i,t}(\boldsymbol{z}_{i,t}^k))}{\sum_{k\in [K]}\omega_{i,t}^k \exp(-\epsilon l_{i,t}(\boldsymbol{z}_{i,t}^k))}
\end{gather}
\ENDFOR
\end{algorithmic}
\end{algorithm}  
\subsubsection{Performance Guarantees}
\begin{theorem}
\label{the:nsdrs-drgt}
Select the step sizes according to \eqref{eq:pool}, $\delta_{i,t} = (\frac{|\mathcal{N}_i|F_i}{\tilde{L}})^{\frac{1}{2}}(\frac{R^2/2+RP_T}{t})^{\frac{1}{4}}$, $\beta_{i,t} = \frac{1}{G_it^{1/2}}$, $\gamma_{i,t} = \frac{1}{G_i^2t^{1/2}}$ and $\xi_{i,t} = \delta_{i,t}/r_i$. Let Assumption~\ref{assump:fg-bound} and \ref{assump:x} hold. By selecting a proper learning rate $\epsilon$ for the meta algorithm, the MA-NSDRS algorithm guarantees the following bound on the expected dynamic regret
\begin{gather}
\mathbb{E}[D-R_T^i] \leq \tilde{\mathcal{O}}(T^{\frac{3}{4}}(1+P_T)^{\frac{1}{2}}),
\end{gather}
for any comparator sequence $\boldsymbol{u}_1, \ldots, \boldsymbol{u}_T \in \mathcal{X}$. Besides, $\tilde{\mathcal{O}}(\cdot)$ omits $\log T$ factors, and $P_T = \sum_{t=2}^T\norm{\boldsymbol{u}_t - \boldsymbol{u}_{t-1}}$ is the path-length.
\end{theorem}
\begin{proof}
See Appendix~\ref{app:nsdrs-drgt}.
\end{proof}
\begin{corollary}
Follow the hyper parameter selection in Theorem~\ref{the:nsdrs-drgt}: set step sizes according to \eqref{eq:pool}, $\delta_{i,t} = (\frac{|\mathcal{N}_i|F_i}{\tilde{L}})^{\frac{1}{2}}(\frac{R^2/2}{t})^{\frac{1}{4}}$, $\beta_{i,t} = \frac{1}{G_it^{1/2}}$, $\gamma_{i,t} = \frac{1}{G_i^2t^{1/2}}$. Assume that Assumption~\ref{assump:fg-bound} and \ref{assump:x} hold. the static regret of \textbf{Algorithm~\ref{alg:ma-nsdrs}} satisfies
\begin{gather}
S-R_T^i \leq \tilde{\mathcal{O}}(T^{3/4}).
\end{gather}
\end{corollary}
\begin{corollary}
Select the step sizes according to \eqref{eq:pool}, $\delta_{i,t} = (\frac{|\mathcal{N}_i|F_i}{\tilde{L}})^{\frac{1}{2}}(\frac{R^2/2+RP_T}{t})^{\frac{1}{4}}$, $\beta_{i,t} = \frac{1}{G_it^{1/2}}$, $\gamma_{i,t} = \frac{1}{G_i^2t^{1/2}}$ and $\xi_{i,t} = \delta_{i,t}/r_i$. Let Assumption~\ref{assump:fg-bound} and \ref{assump:x} hold. Since the meta algorithm only influences on the update step \eqref{eq:proj-z}, the constraint violation of \textbf{Algorithm~\ref{alg:ma-nsdrs}} remains as stated in Theorem~\ref{the:drs-drgt} in the non-stationary environment
\begin{gather}
    V_T^i \leq \tilde{\mathcal{O}}([(1+P_T)^{\frac{1}{2}}T+(1+P_T)^{\frac{1}{4}}T^{\frac{5}{4}}]^{\frac{1}{2}}).
\end{gather}
Similarly, the constraint violation of \textbf{Algorithm~\ref{alg:ma-nsdrs}} in stationary environment ($P_T = 0$) yields
\begin{gather}
    V_T^i \leq \tilde{\mathcal{O}}(T^{\frac{5}{8}}).
\end{gather}
Besides, when the adjustment step is applied in MA-NSDRS, the dynamic regret bound and the static regret bound follow the same growth-rate as
\begin{gather}
    \mathbb{E}[D-R_T^i] \leq \tilde{\mathcal{O}}(T^{\frac{3}{4}}(1+P_T)^{\frac{1}{2}}),
    S-R_T^i \leq \tilde{\mathcal{O}}(T^{3/4}).
\end{gather}
\end{corollary}
\begin{remark}
DRS algorithm has a storage and computational complexity independent of $t$ since it stores and updates only $\boldsymbol{z}_{i,t}$ and $a_{i,t}$. In contrast, the MA-NSDRS algorithm has a $O(T)$ storage complexity because there are $K = \lceil \frac{1}{2}\log_2(1+T) \rceil + 1$ experts running at each time step. The computational complexity of the MA-NSDRS algorithm is $O(T^2)$. Both the DRS and MA-NSDRS algorithms operate distributedly. However, due to the higher computational costs associated with larger network sizes or longer time horizons, the MA-NSDRS algorithm faces scalability limitations.
\end{remark}
\section{Numerical Results}
\label{sec:experiment}
We evaluate our proposed resource-sharing policy using a real-world dataset, which records the energy production and position of several DER facilities in the state of New York \cite{NY_DERs}.\footnote{The data is available at \url{https://der.nyserda.ny.gov/map/}}~According to our problem formulation, we consider the connected DERs in a relative small graphical area. Therefore, we select facilities in Staten Island. Figure~\ref{fig:ny-ders} shows an overview of the locations of different kinds of DERs facilities. It includes 20-solar photovoltaic system (PV), six combined heat and power (CHP), zero anaerobic digester (ADG), six fuel cells and one energy storage in Staten Island. Here, we only consider the PVs. 

In our setting, each graph vertex is an electricity generator and consumer simultaneously. We assume each node's generation capacity is determined by the dataset. Besides, as the dataset does not include a demand model, we consider a balanced demand across the network. Specifically, the total demand of all nodes equals the total generation capacity, and each node has an equal share of this demand. Besides, we assume that if the distance between two PVs is within two districts, one can consider them connected. Each node can transmit the energy produced to its neighboring nodes as reflected by the graph structure. As our algorithm operates under a bandit feedback model, the specific demand model at each node does not impact its performance in experiments.
\begin{figure}[!ht]
\centering
\includegraphics[width=0.5\linewidth]{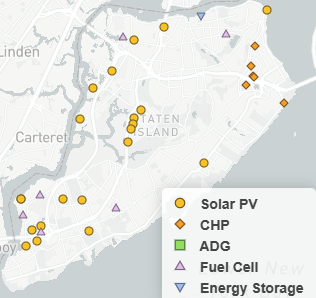}
\caption{Map of DERs facilities.}
\label{fig:ny-ders}
\end{figure}
\subsection{Experiment I}
First, we evaluate our proposed method with a small size network with only six PV facilities (six nodes). Figure~\ref{fig:e1-network} and Figure~\ref{fig:e1-ns-distr} show the network structure and the time-varying generated capacity for each node respectively.
\begin{figure}[!ht]
\begin{subfigure}{.20\textwidth}
 \centering
  \vspace{-10pt}
  \includegraphics[width=.99\linewidth]{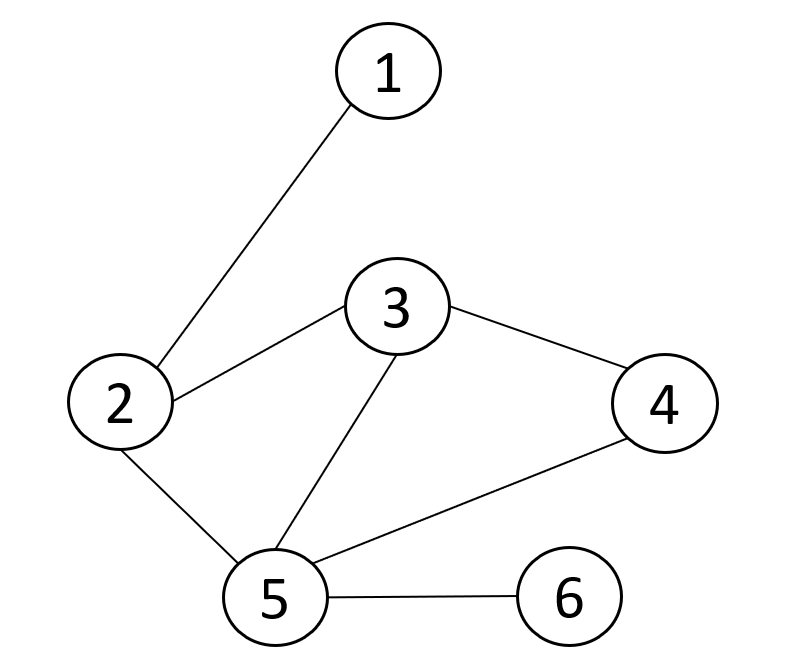}
  \vspace{10pt}
  \caption{Network structure.}
  \label{fig:e1-network}
\end{subfigure}
\begin{subfigure}{.27\textwidth}
  \centering
  \includegraphics[width=.99\linewidth]{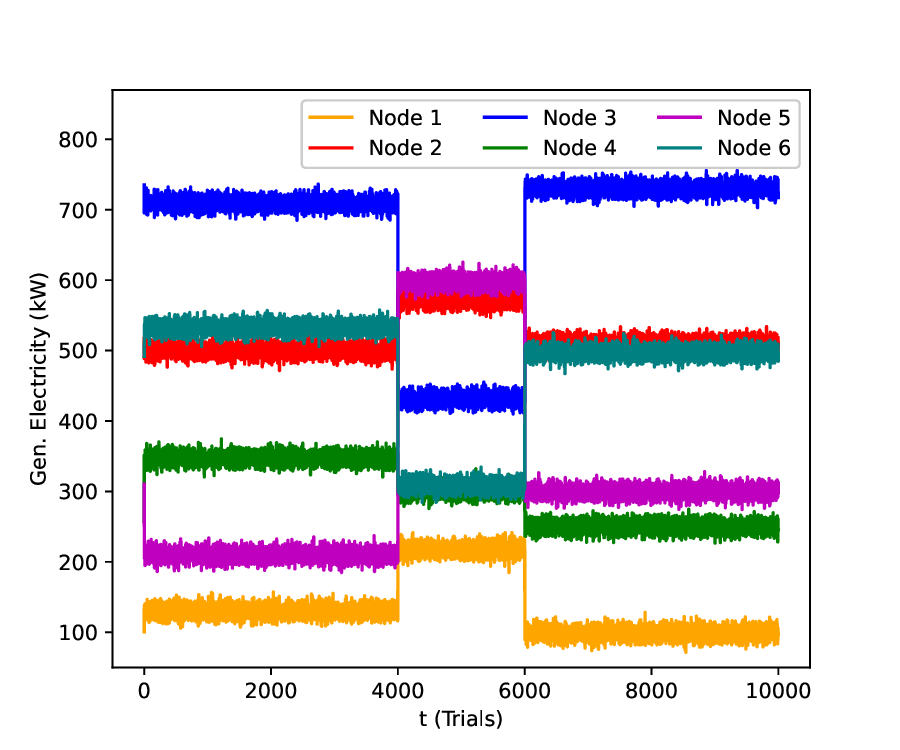} 
  \caption{Time-varying electricity generation.}
  \label{fig:e1-ns-distr}
\end{subfigure}
\caption{Setting of Experiment I.}
\label{fig:e1}
\end{figure}

The results appear in Figure~\ref{fig:e1}. Figure~\ref{fig:e1-rgt} shows the system's cumulative loss. Figure~\ref{fig:e1-constr} depicts the corresponding constraint violation. In addition, the figures illustrate the difference between the two cases when the adjustment is utilized and otherwise. 
\begin{figure}[!ht]
\begin{subfigure}{.24\textwidth}
 \centering
  \includegraphics[width=.99\linewidth]{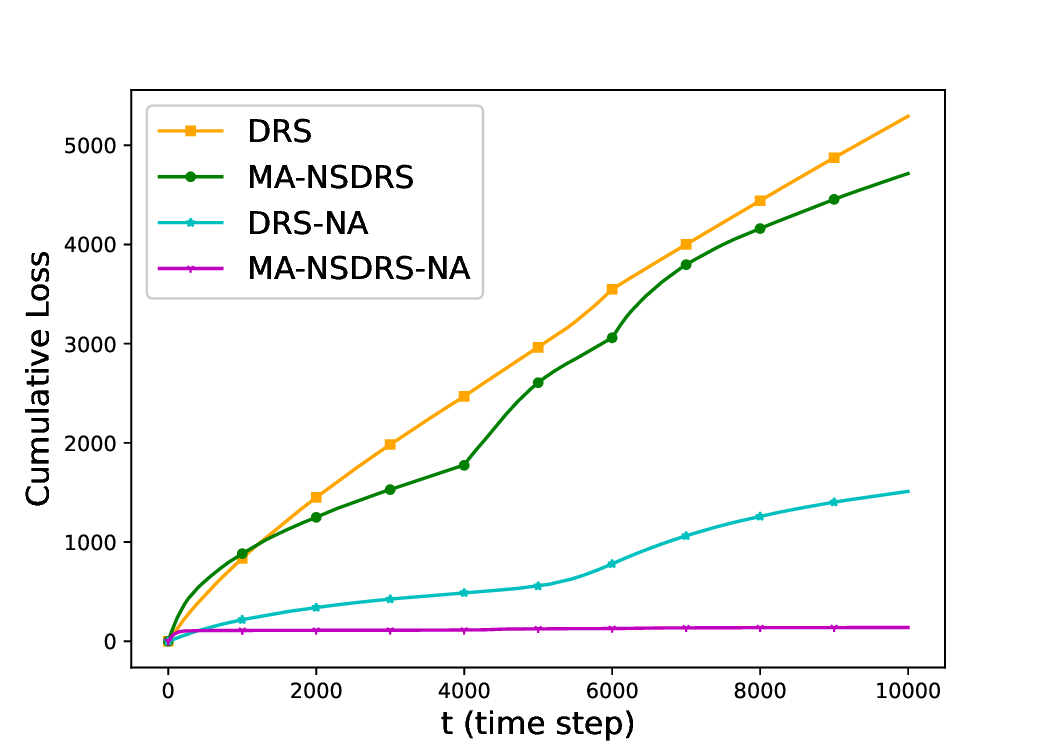}
  \caption{Average cumulative loss.}
  \label{fig:e1-rgt}
\end{subfigure}
\begin{subfigure}{.24\textwidth}
  \centering
  \includegraphics[width=.99\linewidth]{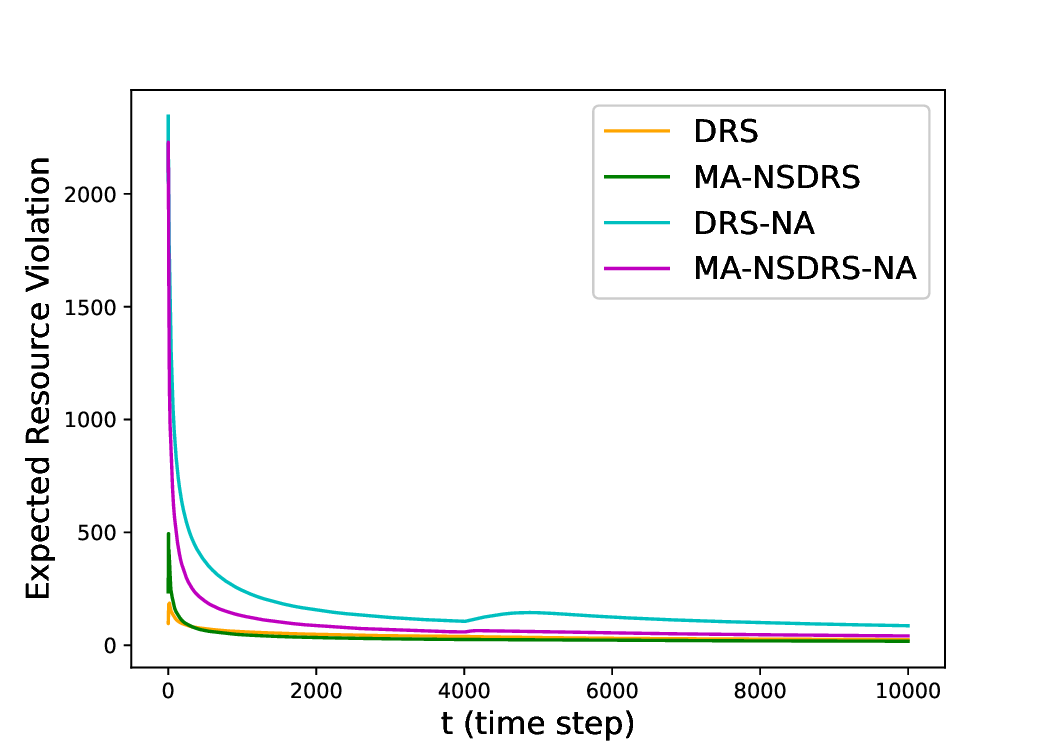} 
  \caption{Constraint violation.}
  \label{fig:e1-constr}
\end{subfigure}
\caption{Result of Experiment I.}
\label{fig:e1}
\end{figure}
Figure~\ref{fig:e1} shows that MA-NSDRS-NA, which is our proposed Algorithm~\ref{alg:ma-nsdrs} without the adjustment step, has the best performance concerning both regret and constraint violation. Indeed, it has the lowest cumulative loss and its long term constraint violation are also guaranteed to converge to zero.

Compared to DRS algorithms, MA-NSDRS algorithms have better performance; that is, the meta algorithm in MA-NSDRS can track the optimal hyperparameter in the non-stationary environment. 

Besides, after adding the adjustment step, the regret of both algorithms (DRS and MA-NSDRS) increases, which is consistent with our theoretical analysis. Figure~\ref{fig:e1-constr} shows that the adjustment step reduces the constraint violation to zero, whereas the algorithms without the adjustment step offer only long-term guarantees. Thus, there is a trade-off between minimizing the regret and constraint violation. To evaluate our algorithm from the energy-efficient perspective, in Figure~\ref{fig:e1-res}, we depict the percentage of the energy level compared with the required amount in each node. We also include the initial energy level as the reference. 
\begin{figure}[!ht]
\centering
\includegraphics[width=0.9\linewidth]{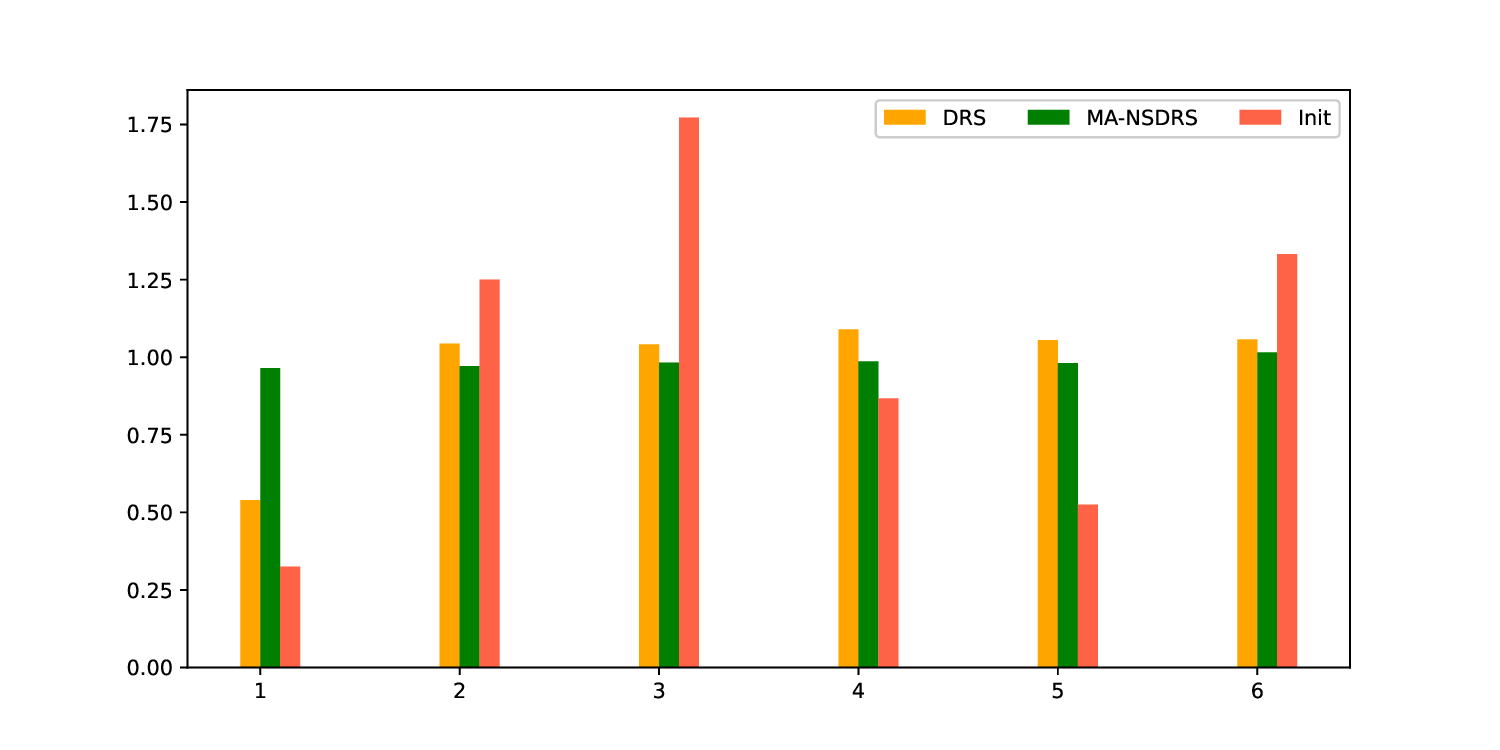}
\caption{Resource distribution of each node.}
\label{fig:e1-res}
\end{figure}

From Figure~\ref{fig:e1-res}, one can conclude that our proposed algorithms, especially MA-NSDRS, result in an efficient energy distribution among nodes, meaning that most of them achieve their required level. Particularly, the proposed algorithms change the situation of nodes 1, 4, and 5 from energy-inadequate to full satisfaction, i.e., above the threshold 1.0. At the same time, they prevent wasting the energy in node 2,3,6.  
\subsection{Experiment II}
In this section, we evaluate our proposed algorithms in a network with 20 PV facilities. Figure~\ref{fig:e2-network} shows the network structure. Besides, Table~\ref{tab:e2-pt} includes the path length $P_T$ of the optimal comparator sequence under the time-varying energy distribution. We compare our proposed algorithms with one benchmark, i.e., an adapted version of BanSaP algorithm proposed in \cite{chen2018bandit}. The BanSaP algorithm addresses the challenge of online convex optimization with time-varying loss functions and constraints, aligning well with our problem formulation. BanSaP algorithm has the same computational and storage complexity as our proposed DRS algorithm. Besides, it was specifically designed to tackle the task allocation problem within IoT management, a scenario analogous to our energy resource allocation challenge. Therefore, we choose to benchmark our proposed algorithm against BanSaP to evaluate its effectiveness.
\begin{figure}[!ht]
\centering
\includegraphics[width=0.5\linewidth]{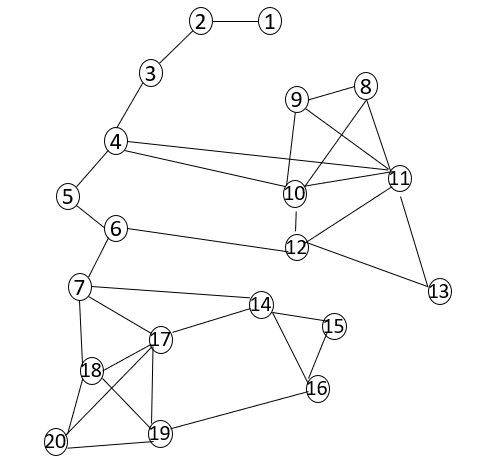}
\caption{Network structure of Experiment II.}
\label{fig:e2-network}
\end{figure}
\begin{table}[!ht]
    \centering
    \begin{tabular}{|c|c|c|c|c|c|}
    \hline
    No. & $P_T$ & No. & $P_T$ & No. & $P_T$\\
    \hline \hline
    1 & 974.44 & 8 & 701.17 & 15 & 1008.00\\
    \hline
    2 & 856.72 & 9 & 520.08 & 16 & 724.83 \\
    \hline
    3 & 1272.47 & 10 & 959.51 & 17 & 697.88\\
    \hline
    4 & 1005.35 & 11 & 895.04 & 18 & 644.84\\
    \hline
    5 & 1372.40 & 12 & 2227.02 & 19 & 460.73\\
    \hline
    6 & 1255.26  & 13 & 846.01 & 20 & 499.80\\
    \hline 
    7 & 804.10 & 14 & 1400.75 & & \\
    \hline
    \end{tabular}
    \caption{$P_T$ of each node under the time-varying energy distribution.}
    \label{tab:e2-pt}
\end{table}
\begin{figure}[!ht]
\begin{subfigure}{.24\textwidth}
 \centering
  \includegraphics[width=.99\linewidth]{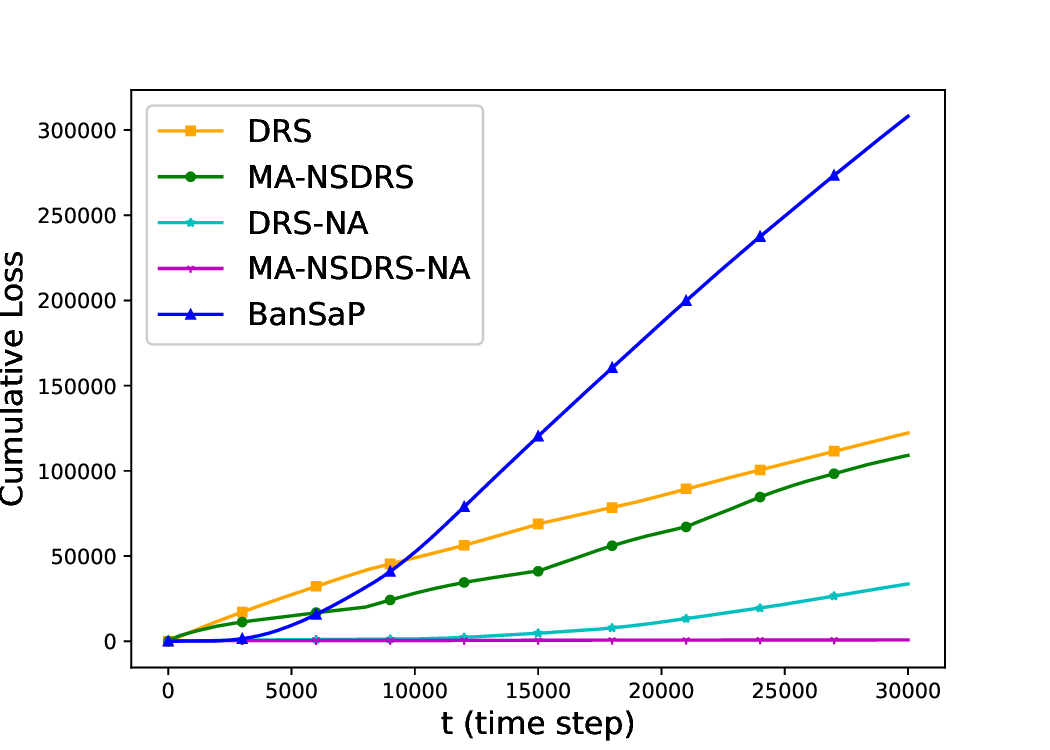}
  \caption{Average cumulative loss.}
  \label{fig:e2-rgt}
\end{subfigure}
\begin{subfigure}{.24\textwidth}
  \centering
  \includegraphics[width=.99\linewidth]{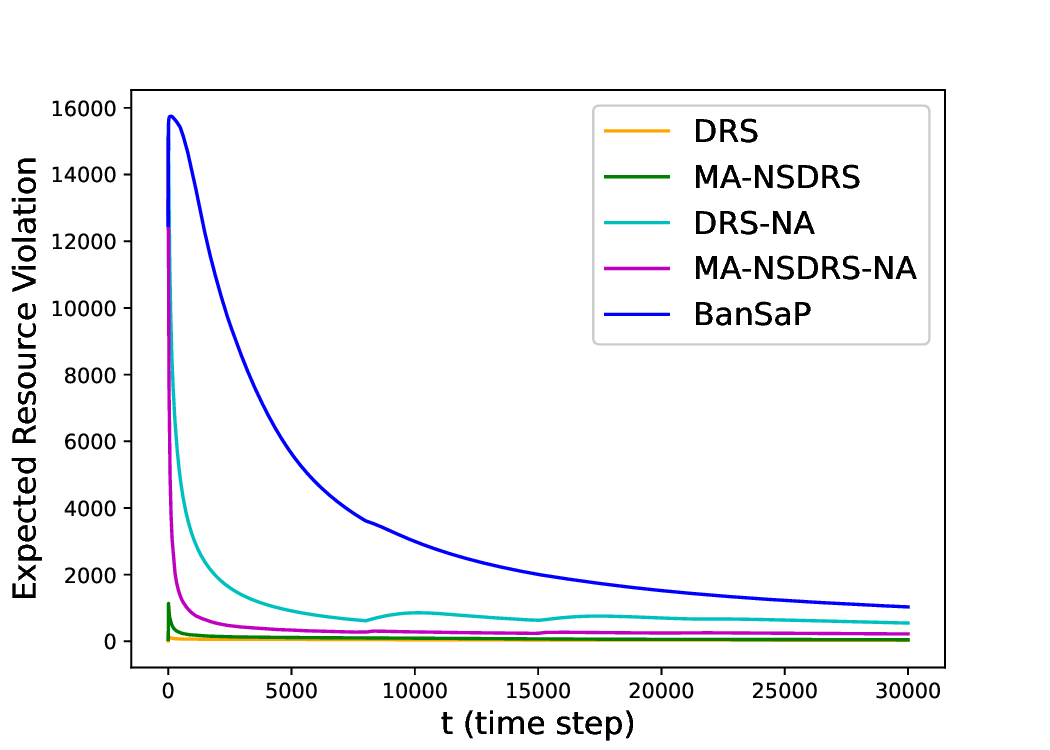} 
  \caption{Constraint violation.}
  \label{fig:e2-constr}
\end{subfigure}
\caption{Result of Experiment II.}
\label{fig:e2}
\end{figure}
\begin{figure}[!ht]
\centering
\includegraphics[width=0.99\linewidth]{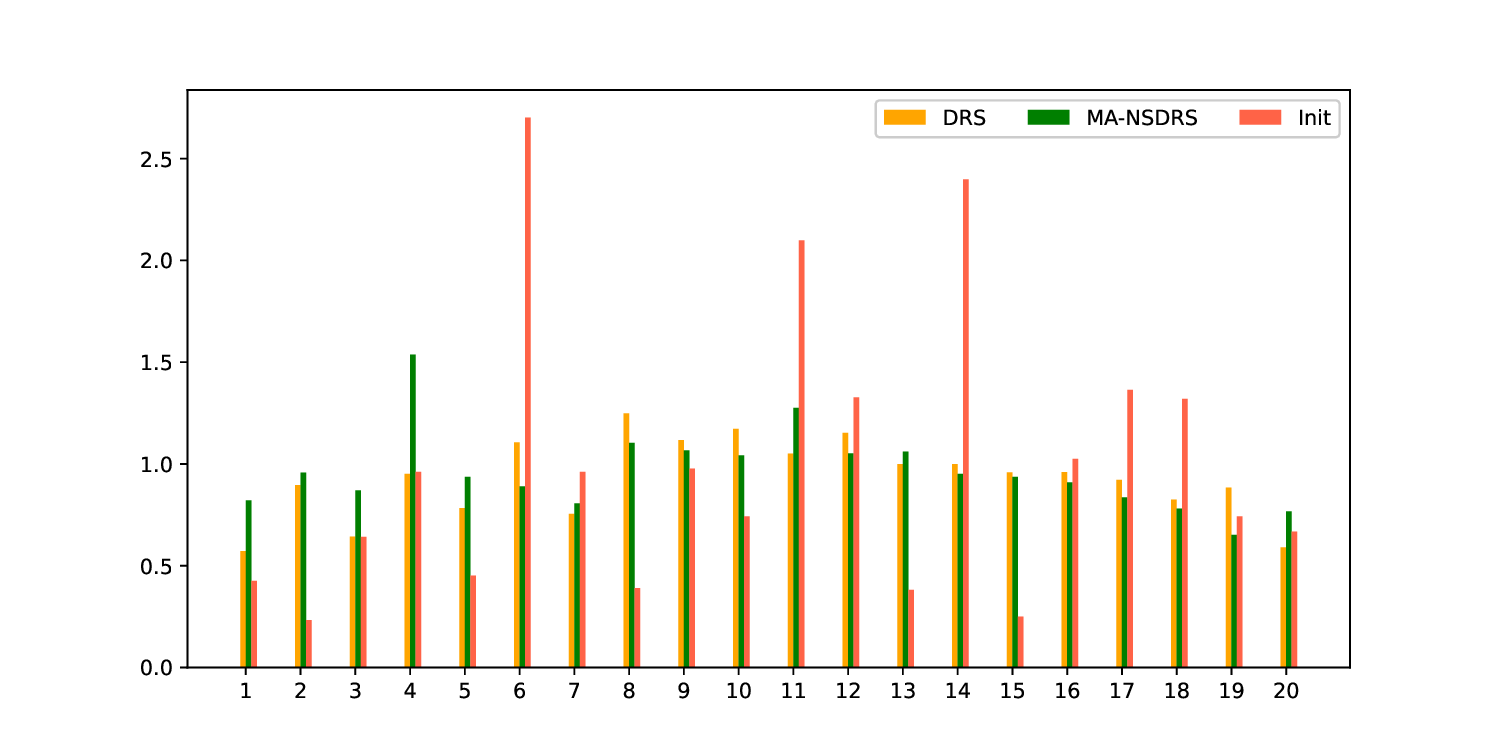}
\caption{Resource distribution of each node.}
\label{fig:e2-res}
\end{figure}

Figure~\ref{fig:e2-rgt} shows the cumulative loss of the proposed algorithms. Figure~\ref{fig:e2-constr} depicts the corresponding constraint violation. Figure~\ref{fig:e2-res} illustrates the resource efficiency level and the initial resource condition. Based on the figures, we can conclude the following.
\begin{itemize}
    \item In Figure~\ref{fig:e2}, our proposed algorithm MA-NSDRS-NA achieves the best performance in terms of minimizing the loss and constraint violation.
    \item Compared to the algorithms without adjustment step, the algorithms with adjustment demonstrate a higher cumulative loss; Nonetheless, when adherence to feasibility constraints is paramount, adjusting the step is inevitable.
    \item Figure~\ref{fig:e2-res} shows the energy level compared with the satisfactory amount in each node after convergence of  the proposed algorithms. Despite the performance dip resulting from the adjustment step, our algorithm excels in reallocating energy resources, ensuring a more equitable distribution across the network. Notably, our proposed algorithm brings about significant enhancements in the energy distribution within the network. Specifically, it effectively uses the energy surplus in nodes 6, 11, and 14 to mitigate the energy shortages in nodes 1, 2, 3, 5, 8, and 13. Remarkably, several shortage nodes(1,2,3,8) have no direct connection to those with excessive energy levels. This achievement shows the considerable advantages of our algorithm in optimizing energy efficiency throughout the entire network.
    \item Compared to BanSaP, all the variations of our proposed algorithm have better performance concerning cumulative loss and constraint violation. There are several reasons to explain the weak performance of BanSaP. First, it uses a constant hyperparameter, which causes high estimation errors in the beginning due to the bandit feedback. Besides, its update step is different from ours. In addition, in BanSaP, there is no adjustment and no meta algorithm to track the optimal hyperparameter. 
    \item Compared to DRS algorithms, MA-NSDRS has better performance in terms of both loss and constraint violation. The reason is the meta algorithm, which tracks the optimal learning rate in the non-stationary environment and thus has a better performance. That implies that our proposed variation of DRS Algorithm~\ref{alg:ma-nsdrs} can manage the dynamic environment better. 
\end{itemize}
\section{Conclusion}
\label{sec:conclusion}
In this study, we have addressed the problem of distributed resource sharing in networks under a dynamic environment. We proposed a distributed algorithm based on the bandit convex optimization and its variation with a meta algorithm. They guarantee $\tilde{\mathcal{O}}((1+P_T)^{\frac{1}{2}}T^{\frac{1}{2}} + (1+P_T)^{\frac{1}{4}}T^{\frac{3}{4}})$ and $\tilde{\mathcal{O}}(T^{\frac{3}{4}}(1+P_T)^{\frac{1}{2}})$ expected regret, respectively. Besides, we evaluate our algorithms using a real-world dataset of DERs and compare their performance with a benchmark, namely, BanSaP. The results show the superior performance of our algorithms. A significant direction for future research is to generalize our methods to more practical settings, for example, by considering the correlation of DERs generation between connected nodes or those nearby each other, considering incorporating power flow models, or allowing for a time-varying/directed network structure.

\appendix
\subsection{Auxiliary Results}
In this section, we provide all auxiliary materials, including definitions, propositions, lemmas, and theorems, which we require to prove our claims.
\begin{definition}\label{def:convexity}\cite{andrei2007convex}
A function $f:\mathbb{R}^n \rightarrow \mathbb{R}$ is convex if its domain is a convex set and for all $\boldsymbol{x}$,$\boldsymbol{y}$ in its domain, and for all $\lambda \in [0,1]$, we have:
\begin{equation}
f(\lambda \boldsymbol{x} +(1-\lambda)\boldsymbol{y}) \leq \lambda f(\boldsymbol{x}) + (1- \lambda) f(\boldsymbol{y}).
\end{equation}
\end{definition}
\begin{definition}\label{def:lipschitz}\cite{heinonen2005lectures}
A real-valued function $f:\mathbb{R} \rightarrow \mathbb{R}$ is called Lipschitz continuous if there exists a positive real constant K such that, for all real $x$ and $y$,
\begin{equation}
\Vert f(x) - f(y) \Vert \leq K \Vert x-y \Vert
\end{equation}
\end{definition}

\begin{proposition}\cite{zhao2021bandit}
\label{prop:sur-loss}
$\mathbb{E}[\hat{f}_t(\boldsymbol{z}_t) - \hat{f}_t(\boldsymbol{v}_t)] \leq \mathbb{E}[\langle \hat{\nabla}f_{t}(\boldsymbol{x}), \boldsymbol{z}_t - \boldsymbol{v}_t \rangle]$. 
\end{proposition}
\begin{lemma} \cite{Yi2020distr_band_online_conv_opt}
\label{lem:ine}
For any constants $\theta \in [0,1]$, $\kappa \in [0,1)$, and $s \leq T \in \mathbb{N}_{+}$, it holds that 
\begin{align}
    &(t+1)^{\kappa}(\frac{1}{t^\theta} - \frac{1}{(t+1)^\theta}) \leq \frac{1}{t} \quad \forall t \in \mathbb{N}_{+}, \\
    &\sum_{t=s}^T \frac{1}{t^{\kappa}} \leq \frac{T^{1-\kappa}}{1-\kappa}, \\
    &\sum_{t=s}^T \frac{1}{t} \leq 2 \log T, \quad \text{if} \quad T \geq 3.
\end{align}
\end{lemma}
Based on Theorem 2 in \cite{zinkevich2003online}, we can conclude the dynamic regret of online gradient descent (OGD) with time-varying learning rate as shown in the following Theorem. 
\begin{theorem}[Dynamic Regret of OGD]
\label{the:ogd} 
Consider the online gradient descent (OGD), which starts with any $\boldsymbol{x} \in \mathcal{X}$ and performs the following update  \cite{zinkevich2003online,zhao2021bandit} 
\begin{gather*}
\boldsymbol{x}_{t+1} = \mathcal{P}_{\mathcal{X}}(\boldsymbol{x}_t - \eta_t \nabla f_t(\boldsymbol{x}_t)).
\end{gather*}
Assume $\mathcal{X}$ is bounded, i.e., $\norm{\boldsymbol{x} - \boldsymbol{y}}_2 \leq R$, $\forall \boldsymbol{x}, \boldsymbol{y} \in \mathcal{X}$ and the gradient is also bounded $\norm{\nabla f_t(\boldsymbol{x_t})} \leq G_f$, $\forall \boldsymbol{x} \in \mathcal{X}, t \in [T]$. The dynamic OGD is upper bounded by 
\begin{align*}
\sum_{t=1}^T f_t(\boldsymbol{x}_t)  &- \sum_{t=1}^T f_t(\boldsymbol{u}_t) \leq  G_fF^2\sum_{t=1}^T \eta_t + \frac{R^2}{2\eta_1} \\
& + R^2\sum_{t=1}^T(\frac{1}{2\eta_t} - \frac{1}{2\eta_{t+1}}) + R\sum_{t=1}^T\frac{\norm{u_{t+1} - u_t}}{\eta_t},
\end{align*}
for any comparator sequence $\boldsymbol{u}_1, \ldots, \boldsymbol{u}_T \in \mathcal{X}$. Besides, $P_T = \sum_{t=2}^T\norm{\boldsymbol{u}_t - \boldsymbol{u}_{t-1}}$ is the path-length.
\end{theorem}
Note that the proofs of Lemma~\ref{lem:ine} and Theorem~\ref{the:ogd} stated above appear in \cite{Yi2020distr_band_online_conv_opt,zinkevich2003online,zhao2021bandit}, so we do not include them here to avoid redundancy.
\begin{lemma}\label{lem:step}
Let $\{\boldsymbol{x}_t\}$ be the decision sequence generated by \eqref{eq:ogd}. The following inequality holds for any sequence $\{\boldsymbol{u}_t\}$ with $\boldsymbol{u}_t \in \mathcal{X}$, $\forall t$,
\begin{align}
    &\mathcal{L}_t(\boldsymbol{x}_t,q) - \mathcal{L}_t(\boldsymbol{u}_t,q_t) \notag \\
    &\leq \frac{1}{2\eta_t}(\norm{\boldsymbol{u}_t - \boldsymbol{x}_t}^2 - \norm{\boldsymbol{u}_t - \boldsymbol{x}_{t+1}}^2) \notag \\
    & \quad + \frac{1}{2\gamma_t} (\norm{q - q_t}^2 - \norm{q- q_{t+1}}^2) \notag \\
    & \quad + \frac{\eta_t \norm{\nabla_x\mathcal{L}_t(\boldsymbol{x}_t,q_t)}^2}{2} + \frac{\gamma_t \norm{\nabla_q\mathcal{L}_t(\boldsymbol{x}_t,q_t)}^2}{2}.
\end{align}
\end{lemma}
\begin{proof}
According to the update step in our algorithm, we have 
\begin{align}
    &\norm{\boldsymbol{u}_t - \boldsymbol{x}_{t+1}}^2  \notag \\
    &= \norm{\boldsymbol{u}_t - \mathcal{P}_{\mathcal{X}}(\boldsymbol{x})} \notag \\
    &\leq  \norm{\boldsymbol{u}_t - (\boldsymbol{x}_t - \eta_t \nabla_x \mathcal{L}_t(\boldsymbol{x}_t,q_t))}^2 \notag \\
    &= \norm{\boldsymbol{u}_t - \boldsymbol{x}_t}^2 + 2(\boldsymbol{u}_t - \boldsymbol{x}_t)^T\eta_t\nabla_x \mathcal{L}_t(\boldsymbol{x}_t,q_t) + \eta_t^2\norm{\nabla_x \mathcal{L}_t(\boldsymbol{x}_t,q_t)}^2.     
\end{align}
Thus, we obtain
\begin{align}
    (\boldsymbol{x}_t - \boldsymbol{u}_t)^T\nabla_x \mathcal{L}_t(\boldsymbol{x}_t,q_t) &\leq \frac{1}{2\eta_t}(\norm{\boldsymbol{u}_t - \boldsymbol{x}_t}^2 - \norm{\boldsymbol{u}_t - \boldsymbol{x}_{t+1}}^2) \notag \\
    & \quad \quad +\frac{\eta_t}{2}\norm{\nabla_x \mathcal{L}_t(\boldsymbol{x}_t,q_t)}^2.
\end{align}
Similarly, it holds
\begin{align}
    (q - q_t)^T\nabla_q\mathcal{L}_t(\boldsymbol{x}_t,q_t) &\leq \frac{1}{2\gamma_t}(\norm{q - q_t}^2 - \norm{q_t - q_{t+1}}^2) \notag \\
    & \quad \quad +\frac{\gamma_t}{2}\norm{\nabla_q\mathcal{L}_t(\boldsymbol{x}_t,q_t)}^2.
\end{align}
Following \cite{zinkevich2003online}, since the Lagrangian function \eqref{eq:lagr} is convex in its first argument and concave in its second, we arrive at
\begin{align*}
    &\mathcal{L}_t(\boldsymbol{x}_t,q_t) -\mathcal{L}_t(\boldsymbol{u}_t,q_t) \leq (\boldsymbol{x}_t-\boldsymbol{u})^T \nabla_x \mathcal{L}_t(\boldsymbol{u}_t,q_t) \\
    &\mathcal{L}_t(\boldsymbol{x}_t,q)-\mathcal{L}_t(\boldsymbol{x}_t,q_t) \leq (q-q_t) \nabla_q \mathcal{L}_t(\boldsymbol{x}_t,q_t).
\end{align*}
Summing up yields
\begin{align}
    &\mathcal{L}_t(\boldsymbol{x}_t,q) - \mathcal{L}_t(\boldsymbol{u}_t,q_t) \notag \\
    & = \mathcal{L}_t(\boldsymbol{x}_t,q)-\mathcal{L}_t(\boldsymbol{x}_t,q_t) + \mathcal{L}_t(\boldsymbol{x}_t,q_t) -\mathcal{L}_t(\boldsymbol{u}_t,q_t) \notag \\
    &\leq (\boldsymbol{x}_t-\boldsymbol{u})^T \nabla_x \mathcal{L}_t(\boldsymbol{u}_t,q_t) + (q-q_t) \nabla_q\mathcal{L}_t(\boldsymbol{x}_t,q_t) \notag \\
    &\leq \frac{1}{2\eta_t}(\norm{\boldsymbol{u}_t - \boldsymbol{x}_t}^2 - \norm{\boldsymbol{u}_t - \boldsymbol{x}_{t+1}}^2) \notag \\
    & \quad + \frac{1}{2\gamma_t} (\norm{q - q_t}^2 - \norm{q- q_{t+1}}^2) \notag \\
    & \quad + \frac{\eta_t \norm{\nabla_x\mathcal{L}_t(\boldsymbol{x}_t,q_t)}^2}{2} + \frac{\gamma_t \norm{\nabla_q\mathcal{L}_t(\boldsymbol{x}_t,q_t)}^2}{2}.
\end{align}
    
\end{proof}
\begin{theorem}[Expected Dynamic Regret of OGD with constraints]
\label{the:ogdc}
Consider the OGD with constraints, which begins with any $\boldsymbol{x}_1 \in \mathcal{X}$. It performs 
\begin{gather}
    \boldsymbol{x}_{t+1} = \mathcal{P}_{\mathcal{X}}(\boldsymbol{x}_t - \eta_t \nabla_x \mathcal{L}_t(\boldsymbol{x}_t,q_t)) \label{eq:ogd} \\
    q_{t+1} = \mathcal{P}_{(0,+\infty)}(q_t + \gamma_t \nabla_q \mathcal{L}_t(\boldsymbol{x}_t,q_t)) \label{eq:ogd-q}
\end{gather}
Suppose the feasible domain $\mathcal{X}$ is bounded, i.e., $\norm{\boldsymbol{x}-\boldsymbol{y}}_2 \leq R$ for any $\boldsymbol{x}, \boldsymbol{y} \in \mathcal{X}$; Meanwhile, the online functions have bounded gradient magnitude, i.e., $\norm{\nabla f_t}_2 \leq G_f$ and $\norm{\nabla g_t}_2 \leq G_g$ for any $\boldsymbol{x} \in \mathcal{X}$ and $t \in [T]$. Then, the dynamic regret of OGD is upper bounded by 
\begin{align}
    &\sum_{t=1}^T [f_t(\boldsymbol{x}_t) - f_t(\boldsymbol{u}_t)] \notag \\
    &\leq G_fF^2\sum_{t=1}^T \eta_t + G_g G^2 \sum_{t=1}^T \gamma_t + \sum_{t=1}^{T-1} (\frac{1}{2\gamma_{t+1}} - \frac{1}{2\gamma_t})Q^2 \notag \\
    &+ \frac{R^2}{2\eta_1} + R^2\sum_{t=1}^T(\frac{1}{2\eta_t} - \frac{1}{2\eta_{t+1}}) + R\sum_{t=1}^T\frac{\norm{u_{t+1} - u_t}}{\eta_t}
\end{align}
for any comparator sequence $\boldsymbol{u}_1, \ldots, \boldsymbol{u}_T \in \mathcal{X}$ and $P_T$ is the path-length defined as $P_T = \sum_{t=2}^T \norm{\boldsymbol{u}_t - \boldsymbol{u}_{t-1}}_2$. 
\end{theorem}
\begin{proof}
First, we state the key self-bounding property. According to Lemma~\ref{lem:step}, let the $\{\boldsymbol{u}_t\}$ become the optimal sequence, by $g_t(\boldsymbol{u}_t) \leq 0$ and definition of Lagrangian in \eqref{eq:lagr}, we have
\begin{align}
    &f_t(\boldsymbol{x}_t) - f_t(\boldsymbol{u}_t) + q g_t^+(\boldsymbol{x}_t) \notag \\
    &\leq \frac{1}{2\eta_t}(\norm{\boldsymbol{u}_t - \boldsymbol{x}_t}^2 - \norm{\boldsymbol{u}_t - \boldsymbol{x}_{t+1}}^2) \notag \\
    & \quad + \frac{1}{2\gamma_t} (\norm{q - q_t}^2 - \norm{q- q_{t+1}}^2) + \frac{\beta}{2}(\norm{q}^2-\norm{q_t}^2) \notag \\
    & \quad + \frac{\eta_t \norm{\nabla_x\mathcal{L}_t(\boldsymbol{x}_t,q_t)}^2}{2} + \frac{\gamma_t \norm{\nabla_q\mathcal{L}_t(\boldsymbol{x}_t,q_t)}^2}{2}.
\end{align}
Take the summation of the inequality above from $t=1$ to $T$, we have
\begin{align}
&\sum_{t=1}^T [f_t(\boldsymbol{x}_t) - f_t(\boldsymbol{u}_t) + q g_t^+(\boldsymbol{x}_t)] \notag \\
&\leq G_f^2\sum_{t=1}^T\eta_t(F^2+ \norm{q_t}^2)  + G_g^2\sum_{t=1}^T\gamma_t(G^2+\beta_t^2\norm{q_t}^2) \notag \\
& + \sum_{t=1}^{T-1} (\frac{1}{2\gamma_{t+1}} - \frac{1}{2\gamma_t})Q^2 + \sum_{t=1}^T\frac{\beta_t}{2}(\norm{q}^2-\norm{q_t}^2) \notag \\
& + \frac{\norm{q}^2}{2\gamma_1} + \sum_{t=1}^T\frac{1}{2\eta_t}(\norm{\boldsymbol{u}_t - \boldsymbol{x}_t}^2 - \norm{\boldsymbol{u}_t -\boldsymbol{x}_{t+1}}^2).   
\end{align}

This inequality is obtained according to Proposition 3 in \cite{mahdavi2012trading}, $\norm{\nabla_x\mathcal{L}_t(\boldsymbol{x}_t,q_t)}^2 \leq 2G_f^2(1+\norm{q_t}^2)$ and $\norm{\nabla_q \mathcal{L}_t(\boldsymbol{x}_t,q_t)}^2 \leq 2G_g^2(G^2+\beta_t^2\norm{q_t}^2)$, where $F$ is the bound of loss function and $G$ the bound of the constraint function, $G_f$ and $G_g$ are the bound of loss function gradient and constraint function gradient respectively.
Besides,
\begin{align}
    &\sum_{t=1}^T\frac{1}{2\eta_t}(\norm{\boldsymbol{u}_t - \boldsymbol{x}_t}^2 - \norm{\boldsymbol{u}_t -\boldsymbol{x}_{t+1}}^2) \notag \\
    &\leq \sum_{t=1}^T \Big (\frac{1}{2\eta_t}\norm{\boldsymbol{u}_t - \boldsymbol{x}_t}^2 - \frac{1}{2\eta_{t+1}}\norm{\boldsymbol{u}_{t+1} - \boldsymbol{x}_{t+1}}^2 \notag \\
    &+ \frac{1}{2\eta_{t+1}}\norm{\boldsymbol{u}_{t+1} - \boldsymbol{x}_{t+1}}^2 - \frac{1}{2\eta_t}\norm{\boldsymbol{u}_{t+1} - \boldsymbol{x}_{t+1}}^2 \notag \\
    & + \frac{1}{2\eta_t}\norm{\boldsymbol{u}_{t+1} - \boldsymbol{x}_{t+1}}^2 - \frac{1}{2\eta_t}\norm{\boldsymbol{u}_t -\boldsymbol{x}_{t+1}}^2\Big)
\end{align}
According to \cite{guo2022online},
\begin{align*}
    \sum_{t=1}^T(\norm{\boldsymbol{u}_{t+1} - \boldsymbol{x}_{t+1}}^2 - \norm{\boldsymbol{u}_t -\boldsymbol{x}_{t+1}}^2) 
    \leq 2R\sum_{t=1}^T \norm{u_{t+1} - u_t}.
\end{align*}
Therefore, we have
\begin{align*}
    &\sum_{t=1}^T\frac{1}{2\eta_t}(\norm{\boldsymbol{u}_t - \boldsymbol{x}_t}^2 - \norm{\boldsymbol{u}_t -\boldsymbol{x}_{t+1}}^2) \\
    &\leq \frac{R^2}{2\eta_1} + R^2\sum_{t=1}^T(\frac{1}{2\eta_t} - \frac{1}{2\eta_{t+1}}) + R\sum_{t=1}^T\frac{\norm{u_{t+1} - u_t}}{\eta_t}.
\end{align*}
Then we can obtain
\begin{align}
&\sum_{t=1}^T [f_t(\boldsymbol{x}_t) - f_t(\boldsymbol{u}_t) + q g_t^+(\boldsymbol{x}_t)] \notag \\
&\leq G_f^2\sum_{t=1}^T\eta_t(F^2+ \norm{q_t}^2)  + G_g^2\sum_{t=1}^T\gamma_t(G^2+\beta_t^2\norm{q_t}^2)  + \frac{\norm{q}^2}{2\gamma_1} \notag \\
&+ \sum_{t=1}^{T-1} (\frac{1}{2\gamma_{t+1}} - \frac{1}{2\gamma_t})Q^2 + \sum_{t=1}^T\frac{\beta_t}{2}(\norm{q}^2-\norm{q_t}^2) \notag \\
& +  \frac{R^2}{2\eta_1} + R^2\sum_{t=1}^T(\frac{1}{2\eta_t} - \frac{1}{2\eta_{t+1}}) + R\sum_{t=1}^T\frac{\norm{u_{t+1} - u_t}}{\eta_t} .   
\end{align}
Select $\beta_t > 2G_f^2\eta_t + 2G_g^2\gamma_t\beta_t^2$, then the inequality becomes
\begin{align*}
&\sum_{t=1}^T [f_t(\boldsymbol{x}_t) - f_t(\boldsymbol{u}_t) + q g_t^+(\boldsymbol{x}_t)] \notag \\
&\leq G_fF^2\sum_{t=1}^T \eta_t + G_g G^2 \sum_{t=1}^T \gamma_t + \sum_{t=1}^{T-1} (\frac{1}{2\gamma_{t+1}} - \frac{1}{2\gamma_t})Q^2 \\ 
& + \sum_{t=1}^T\frac{\beta_t}{2}\norm{q}^2 
+ \frac{\norm{q}^2}{2\gamma_1} + \frac{R^2}{2\eta_1} + R^2\sum_{t=1}^T(\frac{1}{2\eta_t} - \frac{1}{2\eta_{t+1}}) \\
&+ R\sum_{t=1}^T\frac{\norm{u_{t+1} - u_t}}{\eta_t},
\end{align*}
which can also be written as 
\begin{align*}
&\sum_{t=1}^T [f_t(\boldsymbol{x}_t) - f_t(\boldsymbol{u}_t)] + q \sum_{t=1}^Tg_t^+(\boldsymbol{x}_t)- (\frac{1}{2\gamma_1}+\sum_{t=1}^T\frac{\beta_t}{2})q^2 \\
&\leq G_fF^2\sum_{t=1}^T \eta_t + G_g G^2 \sum_{t=1}^T \gamma_t + \sum_{t=1}^{T-1} (\frac{1}{2\gamma_{t+1}} - \frac{1}{2\gamma_t})Q^2 \\
&+ \frac{R^2}{2\eta_1} + R^2\sum_{t=1}^T(\frac{1}{2\eta_t} - \frac{1}{2\eta_{t+1}}) + R\sum_{t=1}^T\frac{\norm{u_{t+1} - u_t}}{\eta_t}.    
\end{align*}
By taking maximization for $q$ over the range $(0,+\infty)$, we get
\begin{align*}
    &\sum_{t=1}^T [f_t(\boldsymbol{x}_t) - f_t(\boldsymbol{u}_t)] +\frac{[\sum_{t=1}^Tg_t(\boldsymbol{x}_t)^+]^2}{2/\gamma_1 + 2\sum_{t=1}^T\beta_t}\\
    &\leq G_fF^2\sum_{t=1}^T \eta_t + G_g G^2 \sum_{t=1}^T \gamma_t + \sum_{t=1}^{T-1} (\frac{1}{2\gamma_{t+1}} - \frac{1}{2\gamma_t})Q^2 \\
    &+ \frac{R^2}{2\eta_1} + R^2\sum_{t=1}^T(\frac{1}{2\eta_t} - \frac{1}{2\eta_{t+1}}) + R\sum_{t=1}^T\frac{\norm{u_{t+1} - u_t}}{\eta_t}.
\end{align*}
Besides, by substituting the regret bound by its lower bound as $\sum_{t=1}^T f_t(\boldsymbol{x}_t) - f_t(\boldsymbol{u}_t) \geq - FT$, we can obtain the upper bound of constraint violation.
\end{proof}
\subsection{Proof of Proposition~\ref{prop:convex-loss}}
\label{app:convex-loss}
\textbf{Part 1-} At first, we prove the convexity of the loss function \eqref{eq:loss_f}. 

The loss function $f_i(\boldsymbol{x}_i) = 1- \frac{1}{\vert \tilde{\mathcal{N}}_i \vert} \sum_{j \in \tilde{\mathcal{N}}_i}\bar{f}_{j,t} = \frac{1}{\vert \tilde{\mathcal{N}}_i \vert} \sum_{j \in \tilde{\mathcal{N}}_i} \Big (1 - \min\{\frac{\sum_{k \in \tilde{\mathcal{N}}_j} x_{k}(j)}{l_i} , 1\} \Big)$ is the combination of (1- $\bar{f})$. Thus, it suffices that we prove the convexity of function $\bar{f}$ here. 
\begin{align*}
1 - \bar{f}_{j,t} &= 1 - \min\{\frac{\sum_{k \in \tilde{\mathcal{N}}_j} x_{k}(j)}{l_j} , 1\} \\
&=1 - \min\{\frac{(\sum_{k \in \mathcal{N}_j}x_{k}(j))+x_i(j)}{l_j} , 1\}  \\
&\overset{\mathrm{(1)}}{=} 1 -\min\{\frac{q(j)+x_i(j)}{l_j} , 1\}  \\
&\overset{\mathrm{(2)}}{=} 1 -\frac{1}{2}(\frac{q(j)+x_i(j)}{l_j} + 1 -\vert \frac{q(j)+x_i(j)}{l_j} - 1\vert)  \\
&=1 -\frac{1}{2l_j}\sum_{j \in \tilde{\mathcal{N}}_i}(q(j)+x_i(j)+l_j -\vert q(j)+x_i(j)-l_j \vert) \\ 
&=1 -(q(j)+l_j)-\frac{1}{2l_j}(x_i(j) -\vert q(j)+x_i(j)-l_j \vert) \\
&\overset{\mathrm{(3)}}{=}C-\frac{1}{2l_j}(x_i(j) -\vert q(j)+x_i(j)-l_j \vert).
\end{align*}
For simplicity, in Equality (1), we have $q(j) = \sum_{k \in \mathcal{N}_j}x_{k}(j)$. Equality (2) refers to analyzing the minimum value of two numbers based on the following formula $\min(a,b)=\frac{1}{2}(a+b-\vert a-b \vert)$. Equality (3) refers to setting $C = 1 -\frac{1}{2l_j}(q(j)+l_j)$. Thus, according to Definition~\ref{def:convexity}, for the function to be convex, the following must hold:
\begin{align}
& C-\frac{1}{2l_j}(\lambda x_i(j)+(1-\lambda)y_i(j)
 \notag \\
& \quad -\vert q(j)+\lambda x_i(j)+(1-\lambda)y_i(j)-l_j \vert) \leq \nonumber \\  
& \lambda \Big (C-\frac{1}{2l_j}(x_i(j) -\vert q(j)+x_i(j)-l_j \vert) \Big ) +(1-\lambda) \Big ( C \notag \\
& \quad -\frac{1}{2l_j}(y_i(j) -\vert q(j)+y_i(j)-l_j \vert) \Big ) \nonumber \\ 
&\xLeftrightarrow{\text (4) } (\lambda x_i(j)+(1-\lambda)y_i(j)) \notag \\
& \quad -\vert q(j)+(\lambda x_i(j)+(1-\lambda)y_i(j))-l_j \vert) \geq \notag \\
&\lambda (x_i(j) -\vert q(j)+x_i(j)-l_j \vert) \notag \\
& \quad +(1-\lambda)(y_i(j) -\vert q(j)+y_i(j)-l_j\vert)  \nonumber \\ 
& \xLeftrightarrow{\text (5) } - \vert q(j)+(\lambda x_i(j)+(1-\lambda)y_i(j))-l_j \vert) \notag \\
& \quad \geq - \lambda \vert q(j)+x_i(j)-l_j \vert - (1-\lambda) \vert q(j)+y_i(j)-l_j \vert \notag \\ 
&  \xLeftrightarrow{}  \vert q(j)+(\lambda x_i(j)+(1-\lambda)y_i(j))-l_j \vert) \notag \\
& \quad \leq \lambda \vert q(j)+x_i(j)-l_j \vert + (1-\lambda) \vert q(j)+y_i(j)-l_j \vert, \label{eq1}
\end{align}
where (4) follows by setting $C = \lambda C+(1-\lambda)C$ and multiplying both sides of the equation by $-\frac{1}{2l_j}$. (5) results from $\lambda x_i(j)+(1-\lambda) y_i(j)= (\lambda x_i(j)+(1-\lambda)y_i(j))$. Using the triangle inequality, i.e., $\vert\gamma+\delta\vert \leq \vert \gamma \vert + \vert \delta \vert$, one can prove that the following inequality holds:
\begin{align*}
\vert q(j) &+(\lambda x_i(j)+(1-\lambda)y_i(j))-l_j \vert) \\
& \leq \lambda \vert q(j)+x_i(j)-l_j \vert + (1-\lambda) \vert q(j)+y_i(j)-l_j \vert
\end{align*}
Finally, by summing over all the neighbors, we conclude that our loss function is convex.

\textbf{Part 2-} Secondly, we prove the Lipschitz-continuity of the loss function \eqref{eq:loss_f}.

Based on definition, the loss function is
\begin{align*}
    f_{i,t}(\boldsymbol{x}_{i,t}) = 1- \frac{1}{\vert \tilde{\mathcal{N}}_i \vert}\sum_{j \in \tilde{\mathcal{N}}_i} \min\{\frac{\sum_{k \in \mathcal{N}_j }x_{k,t}(j) + x_{i,t}(j)}{l_j}, 1\}, 
\end{align*}
where $\sum_{k \in \mathcal{N}_j}x_{k,t}(j)$ only depends on the action of node $i$'s neighbors. Denote $\sum_{k \in \mathcal{N}_j}x_{k,t}(j) = m$, for different actions $\boldsymbol{x}_{i,t}$ and $\boldsymbol{y}_{i,t}$ of node $i$,
\begin{align}
    &f(\boldsymbol{x}_{i,t}) - f(\boldsymbol{y}_{i,t}) \notag \\
    & = \frac{\sum_{j \in \tilde{\mathcal{N}}_i} }{\vert \tilde{\mathcal{N}}_i \vert}\Big(\min \{\frac{m + y_{i,t}(j)}{l_j} , 1\} - \min\{\frac{m + x_{i,t}(j)}{l_j} , 1\} \Big) \label{eq:fx-fy},
\end{align}
which leads to four possible result
\begin{itemize}
    \item $\frac{m + y_{i,t}(j)}{l_j} < 1$ and $\frac{m + x_{i,t}(j)}{l_j} < 1$: $f(\boldsymbol{x}_{i,t}) - f(\boldsymbol{y}_{i,t}) = \sum_{j \in \tilde{\mathcal{N}}_i} \frac{y_{i,t}(j)-x_{i,t}(j)}{l_j}$;
    \item $\frac{m + y_{i,t}(j)}{l_j} \geq 1$ and $\frac{m + x_{i,t}(j)}{l_j} \geq 1$: $f(\boldsymbol{x}_{i,t}) - f(\boldsymbol{y}_{i,t}) = 0$;
    \item $\frac{m + y_{i,t}(j)}{l_j} \geq 1$ and $\frac{m + x_{i,t}(j)}{l_j} < 1$: $f(\boldsymbol{x}_{i,t}) - f(\boldsymbol{y}_{i,t}) = \sum_{j \in \tilde{\mathcal{N}}_i} (1-\frac{x_{i,t}(j)}{l_j}) \leq \sum_{j \in \tilde{\mathcal{N}}_i} \frac{y_{i,t}(j)-x_{i,t}(j)}{l_j}$;
    \item $\frac{m + y_{i,t}(j)}{l_j} < 1$ and $\frac{m + x_{i,t}(j)}{l_j} \geq 1$: $\sum_{j \in \tilde{\mathcal{N}}_i} \frac{y_{i,t}(j)-x_{i,t}(j)}{l_j} \leq f(\boldsymbol{x}_{i,t}) - f(\boldsymbol{y}_{i,t}) = \sum_{j \in \tilde{\mathcal{N}}_i} (\frac{y_{i,t}(j)}{l_j}-1) < 0$. 
\end{itemize}
Thus, for different actions $\boldsymbol{x}_{i,t}$ and $\boldsymbol{y}_{i,t}$ of node $i$, we have $\Vert f(x) - f(y) \Vert \leq K \Vert x-y \Vert$, where $K \geq \frac{1}{\min_{j \in \mathcal{N}} l_j}$. According to Definition~\ref{def:lipschitz}, the loss function is Lipschitz continuous with Lipschitz constant $L = \frac{1}{\min_{j \in \mathcal{N}} l_j}$.

\subsection{Proof of Theorem~\ref{the:drs-drgt}}
\label{app:drs-drgt}
The expected dynamic regret can be decomposed as:
\begin{align}
    & \quad \mathbb{E}[\sum_{t=1}^T f_{i,t}(\boldsymbol{x}_{i,t})] - \sum_{t=1}^T f_{i,t}(\boldsymbol{u}_{i,t}) \notag \\
    &= \underbrace{\mathbb{E}[\sum_{t=1}^T \Big( \hat{f}_{i,t}(\boldsymbol{z}_{i,t}) - \hat{f}_{i,t}(\boldsymbol{v}_{i,t})\Big)]}_{\text{term (a)}} \notag \\
    &\quad + \underbrace{\mathbb{E}[\sum_{t=1}^T \Big( f_{i,t}(\boldsymbol{x}_{i,t}) - \hat{f}_{i,t}(\boldsymbol{z}_{i,t})\Big)]}_{\text{term (b)}} \notag \\
    &\quad +  \underbrace{\mathbb{E}[\sum_{t=1}^T \Big( \hat{f}_{i,t}(\boldsymbol{v}_{i,t}) - f_{i,t}(\boldsymbol{u}_{i,t})\Big)]}_{\text{term (c)}}, 
    \label{eq:dec-drgt}
\end{align}
where $\boldsymbol{v}_{i,1}, \ldots, \boldsymbol{v}_{i,T}$ is the comparator sequence, $\boldsymbol{v}_{i,t} = (1-\xi_{i,t})\boldsymbol{u}_{i,t}$, and $\xi_{i,t}$ is the shrinkage parameter.

The term (a) is essentially the dynamic regret of the smoothed functions. In the bandit feedback model, the gradient estimator is set according to \eqref{eq:est-grad}. Therefore, the step \eqref{eq:ogd} is actually the online gradient descent over the smoothed function $\hat{f}_{i,t}$. Thus, term (a) can be upper bounded by Theorem~\ref{the:ogdc} as follows. 
\begin{align*}
    &\sum_{t=1}^T \Big( \hat{f}_{i,t}(\boldsymbol{z}_{i,t}) - \hat{f}_{i,t}(\boldsymbol{v}_{i,t})\Big)] \\
    &\leq F_i^2\sum_{t=1}^T \hat{G}_{f_{i,t}} \eta_{i,t} +  G_i^2 \sum_{t=1}^T \gamma_t + \sum_{t=1}^{T-1} (\frac{1}{2\gamma_{i,t+1}} - \frac{1}{2\gamma_{i,t}})Q^2  \\
    &+ \frac{R^2}{2\eta_{i,1}} + R^2\sum_{t=1}^T(\frac{1}{2\eta_{i,t}} - \frac{1}{2\eta_{i,t+1}}) + R\sum_{t=1}^T\frac{\norm{\boldsymbol{v}_{i,t+1} - \boldsymbol{v}_{i,t}}}{\eta_{i,t}} \\
    &\leq F_i^2\sum_{t=1}^T 2\frac{|\mathcal{N}_i|^2}{\delta_{i,t}^2} \eta_{i,t}  + G_i^2 \sum_{t=1}^T \gamma_t + \sum_{t=1}^{T-1} (\frac{1}{2\gamma_{i,t+1}} - \frac{1}{2\gamma_{i,t}})Q^2  \\
    &+ \frac{R^2}{2\eta_{i,1}} + R^2\sum_{t=1}^T(\frac{1}{2\eta_{i,t}} - \frac{1}{2\eta_{i,t+1}}) + R\sum_{t=1}^T\frac{\norm{\boldsymbol{v}_{i,t+1} - \boldsymbol{v}_{i,t}}}{\eta_{i,t}}.
\end{align*}
Besides, $\sum_{t=1}^T\norm{\boldsymbol{v}_{i,t+1} - \boldsymbol{v}_{i,t}} = \sum_{t=1}^T (1-\xi_{i,t}) \norm{\boldsymbol{u}_{i,t+1} - \boldsymbol{u}_{i,t}}$. The rigorous analysis in \cite{zhao2021bandit} shows that term (b) and term (c) can be bounded by $2L\sum_{t=1}^T\delta_{i,t}$ and $\sum_{t=1}^T(L\delta_{i,t} + LR\xi_{i,t})$ with $L$ the Lipschitz constant of the loss function $f$, respectively, without involving the unknown path-length. Thus, the final regret bound follows as:
\begin{align}
     &\mathbb{E}[\sum_{t=1}^T f_{i,t}(\boldsymbol{x}_{i,t})] - \sum_{t=1}^T f_{i,t}(\boldsymbol{u}_{i,t}) \notag \\
     &\leq F_i^2\sum_{t=1}^T 2\frac{|\mathcal{N}_i|^2}{\delta_{i,t}^2} \eta_{i,t}  + G_i^2 \sum_{t=1}^T \gamma_t + \sum_{t=1}^{T-1} (\frac{1}{2\gamma_{i,t+1}} - \frac{1}{2\gamma_{i,t}})Q^2  \notag \\
     &+ \frac{R^2}{2\eta_{i,1}} + R^2\sum_{t=1}^T(\frac{1}{2\eta_{i,t}} - \frac{1}{2\eta_{i,t+1}}) + R\sum_{t=1}^T\frac{\norm{\boldsymbol{v}_{i,t+1} - \boldsymbol{v}_{i,t}}}{\eta_{i,t}} \notag \\
     &+ 2L\sum_{t=1}^T\delta_{i,t} + \sum_{t=1}^T(L\delta_{i,t} + LR\xi_{i,t}).
\end{align}
By selecting $\delta_{i,t} = (\frac{|\mathcal{N}_i|F_i}{\tilde{L}})^{\frac{1}{2}}(\frac{R^2/2+RP_T}{t})^{\frac{1}{4}}$, $\eta_{i,t} = (\frac{1}{|\mathcal{N}_i|F_i\tilde{L}})^{\frac{1}{2}}(\frac{R^2/2+RP_T}{t})^{\frac{3}{4}}$, $\beta_{i,t} = \frac{1}{G_it^{1/2}}$, $\gamma_{i,t} = \frac{1}{G_i^2t^{1/2}}$, and $\xi_{i,t} = \delta_{i,t}/r_i$, we can obtain the final regret bound:
\begin{align}
    & \quad \mathbb{E}[\sum_{t=1}^T f_{i,t}(\boldsymbol{x}_{i,t})] - \sum_{t=1}^T f_{i,t}(\boldsymbol{u}_{i,t}) \notag \\
    &\leq (|\mathcal{N}_i|F_i\tilde{L})^{\frac{1}{2}}(\sqrt{T}(R^2/2+RP_T)^{\frac{1}{2}} + (R^2/2+RP_T)^{\frac{1}{4}}T^{\frac{3}{4}}) \notag \\
    &\quad + \frac{(|\mathcal{N}_i|F_i\tilde{L})^{\frac{1}{2}}P_T}{R/2+P_T}T^{\frac{3}{4}}((|\mathcal{N}_i|F_i\tilde{L})^{\frac{1}{2}} + R\log T + 1) \notag  \\
    &\quad + Q^2\log T + \sqrt{T} \notag \\
    &= \tilde{\mathcal{O}}((1+P_T)^{\frac{1}{2}}T^{\frac{1}{2}} + (1+P_T)^{\frac{1}{4}}T^{\frac{3}{4}}).
\end{align}
Similarly, the constraint violation can also be decomposed as
\begin{align*}
    \sum_{t=1}^Tg_t(\boldsymbol{x}_t)^+  =  \Big( \sum_{t=1}^Tg_t(\boldsymbol{x}_t)^+ - \sum_{t=1}^T\hat{g}_t(\boldsymbol{z}_t)^+\Big) + \sum_{t=1}^T\hat{g}_t(\boldsymbol{z}_t)^+,
\end{align*}
where the first term is bounded by $\sum_{t=1}^T(L\delta_{i,t} + LR\xi_{i,t})$, and Theorem~\ref{the:ogdc} shows the bound of second term. Thus, the constraint violation is bounded by
\begin{align*}
    &\sum_{t=1}^Tg_t (\boldsymbol{x}_t)^+  \leq \tilde{\mathcal{O}}([(1+P_T)^{\frac{1}{2}}T+(1+P_T)^{\frac{1}{4}}T^{\frac{5}{4}}]^{\frac{1}{2}}).
\end{align*}
%
\subsection{Proof of Corollary~\ref{cor:drs-adj}}\label{app:drs-adj}
\begin{proof}
Based on (\ref{eq:sio}), the total adjustment can be defined as $\Delta = g_{i,t}(\boldsymbol{x}_{i,t})$ with $\Delta = \sum_{i=1}^N |x_{i,t}(i) - x'_{i,t}(i)|$ and because the loss function is Lipschitz continuous, we have
\begin{align*}
    \mathcal{R}'(T) &= R(T) + \sum_{g_{i,t} > 0}[f(\boldsymbol{x}'_{i,t}) - f(\boldsymbol{x}_{i,t})] \notag \\
    &\leq R(T) + L \sum_{g_{i,t} > 0}\norm{\boldsymbol{x}'_{i,t} - \boldsymbol{x}_{i,t}} \notag \\ 
    &\leq  R(T) + L \sum_{g_{i,t} > 0} \sum_{i=1}^N |x_{i,t}(i) - x'_{i,t}(i)| \notag \\
    &= \mathcal{R}(T) + L \sum_{t=1}^T g_{i,t}^+(x_{i,t}) \notag \\
    & \leq \tilde{\mathcal{O}}((1+P_T)^{\frac{1}{2}}T^{\frac{1}{2}} + (1+P_T)^{\frac{1}{4}}T^{\frac{3}{4}}).
\end{align*}
Therefore, the regret increases with the same magnitude as the previous regret bound. That completes the proof.
\end{proof}
\subsection{Proof of Theorem~\ref{the:nsdrs-drgt}}
\label{app:nsdrs-drgt}
As shown in \eqref{eq:dec-drgt}, the dynamic regret can be decomposed into three parts. Based on the analysis in Appendix~\ref{app:drs-drgt}, term (b) and term (c) defined in \eqref{eq:dec-drgt} can be bounded by $2L\sum_{t=1}^T\delta_{i,t}$ and $\sum_{t=1}^T(L\delta_{i,t} + LR\xi_{i,t})$, respectively. 

Proposition~\ref{prop:sur-loss} shows that term (a) defined in \eqref{eq:dec-drgt} can be upper bounded by 
\begin{gather}
    \text{term (a)} \leq \mathbb{E}[\sum_{t=1}^T(l_{i,t}(\boldsymbol{z}_{i,t}) - l_{i,t}(\boldsymbol{v}_{i,t}))].
\end{gather}
Besides, we have \cite{zhao2021bandit}
\begin{align}
    \sum_{t=1}^T(l_t(\boldsymbol{z}_{i,t}) - l_{i,t}(\boldsymbol{v}_{i,t})) 
    &= \underbrace{\sum_{t=1}^Tl_{i,t}(\boldsymbol{z}_{i,t}) - \sum_{t=1}^T l_{i,t}(\boldsymbol{z}_{i,t}^k)}_{\text{meta-regret}} \notag \\
    &+ \underbrace{\sum_{t=1}^T l_{i,t}(\boldsymbol{z}_{i,t}^k) - \sum_{t=1}^T l_{i,t}(\boldsymbol{v}_{i,t}))}_{\text{expert-regret}},
\end{align}
where $\boldsymbol{z}_{i,1}^k, \ldots, \boldsymbol{z}_{i,T}^k$ is the prediction sequence of expert $k$. This regret decomposition works for any $k \in [K]$. Each expert performs deterministic online gradient descent over the surrogate loss. Hence, we use Theorem~\ref{the:ogd} to bound the expert-regret. Assume that $k^*$ is the nearest optimal step size. We have
\begin{align*}
    &\sum_{t=1}^T l_{i,t}(\boldsymbol{z}_{i,t}^{k^*}) - \sum_{t=1}^T l_{i,t}(\boldsymbol{v}_{i,t}) \notag \\
    &\leq  \frac{\hat{G}_f^iF^2}{2}\sum_{t=1}^T \eta_{i,t}^{k^*} + \frac{R_i^2}{2\eta_{i,1}^{k^*}}  + R_i^2\sum_{t=1}^T(\frac{1}{2\eta_{i,t}^{k^*}} - \frac{1}{2\eta_{i, t+1}^{k^*}}) \notag \\
    & \quad + R_i\sum_{t=1}^T\frac{\norm{\boldsymbol{v}_{i,t+1} - \boldsymbol{v}_{i,t}}}{\eta_{i,t}^{\dagger}} \notag \\
    &\leq |\mathcal{N}_i|^2F_i^2 \sum_{t=1}^T \frac{\eta_{i,t}^{\dagger}}{2(\delta_{i,t}^{\dagger})^2}  + \frac{R^2}{\eta_1^{\dagger}}  \notag \\
    &+ R_i^2\sum_{t=1}^T(\frac{1}{\eta_{i,t}^{\dagger}} - \frac{1}{\eta_{i,t+1}^{\dagger}}) + 2R_i\sum_{t=1}^T\frac{\norm{\boldsymbol{v}_{i,t+1} - \boldsymbol{v}_{i,t}}}{\eta_{i,t}^{\dagger}} \notag \\
    &= \sqrt{|\mathcal{N}_i|F_i\tilde{L}}\sqrt{2R_i+R_i P_T}(R_i+1)T^{\frac{3}{4}} + \frac{R_i\sqrt{|\mathcal{N}_i|F_i\tilde{L}}\log T}{\sqrt{2R_i+P_T}},
\end{align*}
where the first inequality follows from Theorem~\ref{the:ogd}, and the second inequality is based on Lemma~\ref{lem:opeg}. Besides, $\eta^{k^*} \leq \eta^{\dagger} \leq 2\eta^{k^*}$ and the last equality holds by substituting $\eta^{\dagger} = \sqrt{R_i(2R_i^2+R_iP_T)/(|\mathcal{N}_i|F_i\tilde{L})}\cdot t^{-3/4}$ and $\delta = \delta^{\dagger} = (|\mathcal{N}_i|F_iR_i/\tilde{L})^{1/2}t^{-1/4}$.

The next step is to bound the meta-regret. The surrogate loss $l_t$ satisfies
\begin{align}
    |l_{i,t}(\boldsymbol{z})| &\leq |\langle \hat{\nabla}f_{i,t}(\boldsymbol{z}_{i,t}), \boldsymbol{z} - \boldsymbol{z}_{i,t} \rangle| \notag \\
    &\leq \norm{\hat{\nabla}f_{i,t}(\boldsymbol{z})}_2 \norm{\boldsymbol{z} - \boldsymbol{z}_{i,t}} \leq 2 \hat{G}_f^iR,
\end{align}
$\forall \boldsymbol{z} \in (1-\xi)\mathcal{X}$. 
\begin{lemma}\cite{zhao2021bandit}
\label{lem:step-size}
For any step size $\epsilon > 0$, we have 
\begin{gather}
\sum_{t=1}^T l_t(\boldsymbol{z}_t) - \min_{k \in [K]}(\sum_{t=1}^Tl_t(\boldsymbol{z}_t^k) + \frac{1}{\epsilon}\ln \frac{1}{\omega_1^k}) \leq 2\epsilon T \hat{G}_f^2 R^2.
\end{gather}
By setting $\epsilon = \sqrt{1/ 2T \hat{G}_f^2 R^2}$, we can obtain that
\begin{gather}
\sum_{t=1}^T l_t(\boldsymbol{z}_t) - \sum_{t=1}^Tl_t(\boldsymbol{z}_t^k) \leq \hat{G}_f R \sqrt{2T}(1+\ln \frac{1}{\omega_1^k})
\end{gather}
where $\hat{G}_f$ is the estimated gradient. 
\end{lemma}
Lemma~\ref{lem:step-size} holds for any $k \in [K]$ including $k^*$. Thus,
\begin{align}
\sum_{t=1}^Tl_{i,t}(\boldsymbol{z}_{i,t}) - &\sum_{t=1}^T l_t(\boldsymbol{z}_{i,t}^{k^*}) \leq \hat{G}_f^i R_i \sqrt{2T}(1+\ln \frac{1}{\omega_{i,1}^{k^*}}) \notag \\
&\leq \frac{R_iF_i|\mathcal{N}_i|}{\delta_{i,T}}\sqrt{2T}(1+ 2\ln(k^*+1)) \notag \\
&=  \sqrt{2R_iF_i|\mathcal{N}_i|\tilde{L}}T^{\frac{3}{4}}(1+ 2\ln(k^*+1)).
\end{align}
By combining with the expert regret and term (b) and (c) in \eqref{eq:dec-drgt}, the dynamic regret is bounded by 
\begin{align}
    & \quad \mathbb{E}[\sum_{t=1}^T f_{i,t}(\boldsymbol{x}_{i,t})] - \sum_{t=1}^T f_{i,t}(\boldsymbol{u}_{i,t}) \notag \\
    &= \text{term (a)} + \text{term (b)} + \text{term (c)} \notag \\
    &\leq \text{term (a)} + 2L\sum_{t=1}^T\delta_{i,t} + \sum_{t=1}^T(L\delta_{i,t} + LR\xi_{i,t}) \notag \\
    &\leq \sqrt{|\mathcal{N}_i|F_i\tilde{L}}\sqrt{2R_i+R_i P_T}(R_i+1)T^{\frac{3}{4}} \notag \\
    &+ \frac{R_i\sqrt{|\mathcal{N}_i|F_i\tilde{L}}\log T}{\sqrt{2R_i+P_T}} + \sqrt{2R_iF_i|\mathcal{N}_i|\tilde{L}}T^{\frac{3}{4}}(1+ 2\ln(k^*+1)) \notag \\
    &\leq \sqrt{|\mathcal{N}_i|F_i\tilde{L}}\sqrt{2R_i+R_i P_T}(R_i+1)T^{\frac{3}{4}} + \frac{R_i\sqrt{|\mathcal{N}_i|F_i\tilde{L}}\log T}{\sqrt{2R_i+P_T}} \notag \\
    &+ \sqrt{2R_iF_i|\mathcal{N}_i|\tilde{L}}T^{\frac{3}{4}}(1+ 2\ln(\lceil \log_2(1+ P_T/(2R_i))\rceil+1)) \notag \\
    &= \tilde{\mathcal{O}}(T^{\frac{3}{4}}(1+P_T)^{\frac{1}{2}}). 
\end{align}
Besides, the update step with respect to the constraint violation keeps the same, thus the constraint violation is the same as Theorem~\ref{the:drs-drgt}. 
\bibliographystyle{IEEEtran}
\bibliography{references}  
\end{document}